\documentclass[12pt]{article}
\usepackage{algorithm,algorithmic,amsmath,amssymb,epsf,epsfig,amsthm,times,ifthen,epic,psfrag,array,etoolbox}
\usepackage{enumerate,subcaption,multirow,makecell,enumitem,tikz,pgfplots}
\usepackage{printtime}
\usepackage{cite}     
\usepackage{fancyhdr}
\usepackage{setspace} 
\usepackage{lastpage} 
\usepackage{url}
\usepackage{hyperref} 

\newcommand{\PartOne}{Part I}                           
\newcommand{\PartTwo}{Part II}                          
\newcommand{\PartOneSectionModel}            {1.1}    
\newcommand{\PartOneCorollaryIdeal}          {2.4}    
\newcommand{\PartOneTheoremSmallerField}     {2.5}    
\newcommand{\PartOneLemmaDirectProduct}      {2.6}    
\newcommand{\PartOneLemmaSubring}            {2.8}    
\newcommand{\PartOneTheoremFieldBeatsRing}   {2.11}   
\newcommand{\PartOneTheoremOnlyThirtyTwo}    {3.10}   
\newcommand{\PartOnePartDomRing}             {5.4}    
\newcommand{\PartOneTheoremBestRingNetwork}  {5.8}    
\newcommand{\PartOneTheoremPFiveSeven}       {5.9}    
\newcommand{\PartOneNChooseTwoFigure}        {1}      

\def\submitteddate{August 5, 2016}
\def\reviseddate{}

\DeclareFontFamily{U}{matha}{\hyphenchar\font45}
\DeclareFontShape{U}{matha}{m}{n}{
      <5> <6> <7> <8> <9> <10> gen * matha
      <10.95> matha10 <12> <14.4> <17.28> <20.74> <24.88> matha12
      }{}
\DeclareSymbolFont{matha}{U}{matha}{m}{n}
\DeclareMathSymbol{\notdivides}{3}{matha}{"1F}
\DeclareMathSymbol{\divides}{3}{matha}{"17}

\begin{document}

\newcommand{\creationtime}{\today \ \ @ \theampmtime}

\pagestyle{fancy}
\renewcommand{\headrulewidth}{0cm}
\chead{\footnotesize{Connelly-Zeger}}
\rhead{\footnotesize{\reviseddate}}
\lhead{}
\cfoot{Page \arabic{page} of \pageref{LastPage}} 

\renewcommand{\qedsymbol}{$\blacksquare$} 


\newtheorem{theorem}              {Theorem}     [section]
\newtheorem{lemma}      [theorem] {Lemma}
\newtheorem{corollary}  [theorem] {Corollary}
\newtheorem{proposition}[theorem] {Proposition}
\newtheorem{remark}     [theorem] {Remark}
\newtheorem{conjecture} [theorem] {Conjecture}
\newtheorem{example}    [theorem] {Example}

\theoremstyle{definition}         
\newtheorem{definition} [theorem] {Definition}
\newtheorem*{claim} {Claim}
\newtheorem*{notation}  {Notation}



\setlist[itemize]{itemsep=0.5pt, topsep=1pt}

\newcommand{\Z}{\mathbf{Z}}
\newcommand{\N}{\mathbf{N}}

\newcommand{\alphabet}{\mathcal{A}}
\newcommand{\A}{\mathcal{A}}
\newcommand{\F}{\mathbb{F}}
\newcommand{\Capacity}{\mathcal{C}}

\newcommand{\GF}[1]{\mathrm{GF}\!\left(#1\right)}
\newcommand{\Char}[1]{\mathsf{char}\!\left(#1\right)}
\newcommand{\Comment}[1]{ \left[\mbox{from  #1} \right]}
\newcommand{\Div}[2]{ #1  \bigm|  #2  } 
\newcommand{\NDiv}[2]{ #1   \notdivides   #2   }
\newcommand{\GCD}[2]{ \mathsf{gcd} \!\left( #1  ,  #2 \right)  } 
\newcommand{\PrimeFact}[1]{p_1^{k_1} \cdots p_{#1}^{k_{#1}}}

\newcommand{\xx}[2]{x_{#1}^{(#2)}}
\newcommand{\ww}[2]{w_{#1}^{(#2)}}

\newcommand{\Module}[2]{{}_{#2} #1}
\newcommand{\act}{\cdot}
\newcommand{\entropy}[1]{H\left(#1\right)}

\newcommand{\osum}{\displaystyle\bigoplus}
\newcommand{\dsum}{\displaystyle\sum}

\newcommand{\DP}{\times}
\newcommand{\bigDP}{\prod}
\newcommand{\newaction}{\odot}

\newcommand{\range}[1]{\mathsf{range}\left(#1 \right)}
\renewcommand{\dim}[1]{\mathsf{dim}\left(#1 \right)}

\renewcommand{\emptyset}{\varnothing} 
\renewcommand{\subset}{\subseteq}     
\newcommand{\Network}{\mathcal{N}}
\newcommand{\TBA}{*** To Be Added ***}

\newcommand{\Fig}[2]{
 \medskip
  \epsfysize=#2 
  \epsffile{#1.eps}
 \medskip
}

\let\bbordermatrix\bordermatrix
\patchcmd{\bbordermatrix}{8.75}{4.75}{}{}
\patchcmd{\bbordermatrix}{\left(}{\left[}{}{}
\patchcmd{\bbordermatrix}{\right)}{\right]}{}{}

\setcounter{page}{0}

\title{Linear Network Coding over Rings \\ \PartTwo: Vector Codes and Non-Commutative Alphabets
\thanks{This work was supported by the 
National Science Foundation.\newline
\indent \textbf{J. Connelly and K. Zeger} are with the 
Department of Electrical and Computer Engineering, 
University of California, San Diego, 
La Jolla, CA 92093-0407 
(j2connelly@ucsd.edu and zeger@ucsd.edu).
}}

\author{Joseph Connelly and Kenneth Zeger\\}

\date{
\textit{
IEEE Transactions on Information Theory\\
Submitted: \submitteddate\\
}}

\maketitle
\begin{abstract}
  We prove the following results regarding the linear solvability of networks over various alphabets.
  For any network,
  the following are equivalent:
  (i) vector linear solvability over some finite field,
  (ii) scalar linear solvability over some ring,
  (iii) linear solvability over some module.
  Analogously, the following are equivalent:
  (a) scalar linear solvability over some finite field,
  (b) scalar linear solvability over some commutative ring,
  (c) linear solvability over some module whose ring is commutative.
  Whenever any network is linearly solvable over a module,
  a smallest such module arises in a vector linear solution for that network over a field.

  If a network is linearly solvable over some non-commutative ring
  but not over any commutative ring,
  then such a non-commutative ring must have size at least $16$,
  and for some networks, this bound is achieved.
  An infinite family of networks is demonstrated,
  each of which is 
  scalar linearly solvable over some non-commutative ring
  but not over any commutative ring.

  Whenever $p$ is prime and $2 \le k \le 6$,
  if a network is scalar linearly solvable over some ring of size $p^k$,
  then it is also $k$-dimensional vector linearly solvable over the field $\GF{p}$,
  but the converse does not necessarily hold.
  This result is extended to all $k\ge 2$ when the ring is commutative.
\end{abstract}

\thispagestyle{empty}

\clearpage

\section{Introduction}
In the companion paper (i.e. \PartOne ) \cite{Connelly-Ring1},
we studied scalar linear network codes over commutative rings.
Equivalently, these are linear codes over modules
where a commutative ring acts on its own additive group via multiplication in the ring.
In particular,
we compared the scalar linear solvability of networks over different types of commutative rings of the same size.
We proved that networks that are scalar linearly solvable over some commutative ring 
are also scalar linearly solvable over some field,
although not necessarily of the same size.
Additionally, we characterized all commutative rings with the property that
there exists a network with a scalar linear solution over the ring
but not over any other commutative ring of the same size.

Linear network codes can be advantageous due to their ease of implementation and mathematical tractability.
These properties are due to the algebraic simplicity of linear maps and also to the
structured nature of the alphabets used.
Fields have the most algebraic constraints among alphabets used for linear network coding,
e.g. associativity, distributivity, commutativity, invertibility.
More generally, rings may lack commutativity and/or invertibility,
thus providing a broader class of alphabets over which to achieve linear network solvability.
We demonstrated in \PartOne{} that 
relaxing only the invertibility constraint 
(i.e. restricting to commutative rings)
can lead to linear network solvability that would not otherwise be possible with fields
of the same alphabet size.

In the present paper (\PartTwo),
we additionally relax the commutativity constraint,
and we study linear coding over general ring alphabets and, even more generally,
over modules.
Vector and scalar linear codes over rings and fields
are special cases of linear codes over modules.
We focus on the relationship between alphabet commutativity
and the scalar and vector linear solvability of networks,
and we compare the linear solvability of networks over different modules 
where the alphabet size is the same.

\subsection{Linear codes over modules} \label{sec:model}

\begin{definition}
  An \textit{$R$-module} (specifically a left $R$-module) is
  an Abelian group
  $(G,\oplus)$ together with a ring%
 \footnote{In this paper we will assume all groups are finite and all rings have a multiplicative identity,
 even when we do not explicitly state these facts.}  
  $(R,+,*)$ of \textit{scalars}
  and an action 
  $$\act : R \times G \to G $$
  such that
  for all $r,s \in R$ and all $g,h \in G$ the following hold:
  \begin{align*}
    r \act (g\oplus h) &= (r \act g) \oplus (r \act h) \\
    (r + s) \act g &= (r \act g) \oplus (s \act g) \\
    (r*s) \act g &= r \act (s \act g) \\
    1 \act g &= g .
  \end{align*}
\end{definition}
For brevity, we will sometimes refer to such an $R$-module 
as $\Module{G}{R}$ or simply $G$.
The \textit{size of a module} will refer to $|G|$.
Let $M_k(R)$ denote the ring of all $k \times k$ matrices with entries in $R$
and let $G^k$ denote the Abelian group of all $k$-dimensional vectors with entries in $G$ 
with vector addition,
where $k$ is a positive integer.
Then $G^k$ is an $M_k(R)$-module where multiplication of elements of $R$
with elements of $G$ is given by the action of $\Module{G}{R}$.
%

For basic network coding definitions, see \PartOne{} \cite[Section \PartOneSectionModel]{Connelly-Ring1}.
We will use the same models as in \PartOne{} for networks, alphabets, etc.,
except we now study the generalized case of linear codes over modules,
as opposed to linear codes over rings.
%
%
An edge function on the out-edge of a network node is \textit{linear with respect to the module $\Module{G}{R}$}
if can be written in the form
\begin{align}
  f(x_1,\dots,x_m) & = (M_1 \act x_1) \oplus \cdots \oplus (M_m \act x_m) 
  \label{eq:a}
\end{align}
where $x_1,\dots,x_m \in G$ are the inputs of the node and
$M_1, \dots, M_m \in R$ are constants.
That is, the messages and edge symbols are elements of the Abelian group $G$,
and the linear edge and decoding functions are determined by coefficients of the ring $R$.
A decoding function is linear with respect to $\Module{G}{R}$
if it has a form analogous to \eqref{eq:a},
and a code is \textit{linear over a module $\Module{G}{R}$} 
if all edge and decoding functions are linear with respect to $\Module{G}{R}$.
The alphabet size in a linear code over a module is the size of the module,
i.e. $|G|$.
The special case of a module where the finite ring $R$ acts on its own Abelian group $(R,+)$
by multiplication in $R$ is denoted by $\Module{R}{R}$,
and in this case, \eqref{eq:a} is equivalent to the definition
of a scalar linear code over a ring that we used in \PartOne.

A network is 
\textit{linearly solvable over a module $\Module{G}{R}$}
if there exists a linear solution over $\Module{G}{R}$.
We will focus on two special types of linear codes:
\begin{itemize}
\item[(i)] A \textit{scalar linear code over a ring $R$}
  is a linear code over the module $\Module{R}{R}$.
  A network is \textit{scalar linearly solvable over $R$} 
  if it has a linear solution over the module $\Module{R}{R}$.
\item[(ii)] A \textit{$k$-dimensional vector linear code over a ring $R$}
  is a linear code over the module $\Module{R^k}{M_k(R)}$.
  A network is \textit{vector linearly solvable over $R$} 
  if it has a linear solution over the module $\Module{R^k}{M_k(R)}$,
  for some positive integer $k$.
\end{itemize}
When referring to a linear code or solution over a ring,
we will always specify (in this paper) scalar versus vector,
or if neither is specified,
then we are referring to a linear code over a module.
Additionally,
when referring to an $R$-module $G$,
the ring $R$ is not assumed to be finite,
unless otherwise specified.
However, when referring to a scalar or vector linear code over a ring $R$,
the ring $R$ is assumed to be finite.

A \textit{$k$-dimensional vector routing code over an alphabet $\A$}
is a code in which
messages and edge symbols are elements of $\A^k$
and edge and decoding functions copy certain input vector components to the certain output vector components.
A vector routing code over $\A$ is, in fact, a special case of a vector linear code over $\A$
where each row of each of the matrices $M_1, \dots, M_m$ in \eqref{eq:a}
is either all zero
or else has $1$ one and $k-1$ zeros,
and for each $i \le k$, at most one of the matrices $M_1, \dots, M_m$
has a non-zero $i$th row.

We can similarly define a right $R$-module
and a linear code over a right $R$-module.
It can easily be shown that any linear code over a right module
is equivalent to a particular linear code over a left module,
so we restrict attention only to left modules.

\subsection{Our contributions}
\label{sec:results}

In Section~\ref{sec:basicmodules},
lemmas are given which are used in proofs later in the paper.

Section~\ref{ssec:simple} analyzes the linear solvability of networks over ring alphabets
which are not necessarily commutative.
In \PartOne,
we proved that whenever a network is scalar linearly solvable over some commutative ring,
then the smallest commutative ring over which the network is scalar linearly solvable is a field 
(and thus the ring is unique)
\cite[Theorem \PartOneTheoremSmallerField]{Connelly-Ring1}.
Here, we prove (in Theorem~\ref{thm:R_dom_by_simple})
that if a network is scalar linearly solvable over some (not necessarily commutative) ring,
then a smallest such ring is a matrix ring over a field.
It remains unknown, however,
whether there can be more than one smallest (not necessarily commutative) ring over which a network is linearly solvable,
since in general, there can exist multiple matrix rings over fields that are the same size.
We demonstrate (in Corollaries~\ref{cor:choose-two-ring} and \ref{cor:dim-n_smallest_ring})
that for two infinite classes of networks studied in this paper,
the smallest size ring over which each network is linearly solvable is indeed unique.

We prove (in Theorem~\ref{thm:min_module}) 
that if a network is linearly solvable over some module,
then a smallest such module (i.e. with a smallest associated Abelian group) 
corresponds to a vector linear solution over some finite field.%
\footnote{For example, in a $k$-dimensional vector linear code over a field $\F$,
the alphabet size of the module is $|\F|^k$.}
We prove 
(in Theorem~\ref{thm:min_module_not_unique}),
in contrast to the commutative ring case,
that the minimum size module
with respect to linear solvability
is not necessarily unique.
Thus, for a fixed network, vector linear codes over fields
are ``best'' in a certain sense,
as these codes can minimize the alphabet size needed for a linear solution.

We also show (in Corollary~\ref{cor:general_vector})
that for all networks, the following properties are equivalent: 
(i) vector linear solvability over some field,
(ii) scalar linear solvability over some ring,
and
(iii) linear solvability over some module.
Similarly, we show (in Corollary~\ref{cor:general_scalar})
that for all networks, the following properties are equivalent: 
(a) scalar linear solvability over some field,
(b) scalar linear solvability over some commutative ring,
and 
(c) linear solvability over some module whose ring is commutative.
%

In Section~\ref{sec:Gen_M},
we present a family of networks that generalize the M Network of \cite{Medard-NonMulticast,DFZ-Matroids},
and we enumerate (in Theorem~\ref{thm:dim-n_vector}) 
the particular vector dimensions over which each of these networks
has vector linear solutions.
We prove 
(in Corollary~\ref{cor:non-comm-ring})
that these networks have scalar linear solutions over certain non-commutative matrix rings
yet do not have scalar linear solutions over any commutative ring.
We also show (in Theorem~\ref{thm:non-comm_16})
that if a network is scalar linearly solvable over a non-commutative ring $R$
and is not scalar linearly solvable over any commutative ring, then $|R| \geq 16$.
This lower bound is shown to be achievable
(in Corollary~\ref{cor:non-comm-ring} and Example~\ref{ex:16-noncommutative})
by exhibiting a network which has a scalar linear solution over a non-commutative ring of size $16$
but not over any commutative ring.

Section~\ref{sec:mods} focuses on linear solvability of networks over different modules with the same alphabet size,
specifically, $k$-dimensional vector linear codes over $\GF{p}$
and rings of size $p^k$.
We prove
(in Theorem~\ref{thm:module-beats-rings})
that for each prime power $p^k$, there exists a network with a linear solution over a
module of size $p^k$ but with no scalar linear solutions over any
ring of size $p^k$.
These particular networks have $k$-dimensional vector linear solutions over $\GF{p}$.
We show (in Theorem~\ref{thm:commutative_implies_vector})
that
any network with a scalar linear solution over a commutative ring of size $p^k$
has a $k$-dimensional vector linear solution over $\GF{p}$.
We prove a similar result (in Theorem~\ref{thm:p6_implies_vector})
for general rings of size $p^k$ when $k \leq 6$.
Additionally, we show 
(in Theorems~\ref{thm:commutative_implies_vector}
and \ref{thm:p6_implies_vector})
that these results generalize in a natural way to rings of non-power-of-prime sizes.

Finally, Section~\ref{sec:complexity} provides some concluding remarks.

\subsection{Comparisons of modules}
\label{sec:basicmodules}

An $R$-module $G$ is \textit{faithful} if 
for all $r \in R\backslash\{0\}$,
there exists $g \in G$
such that $r \act g \ne 0$.
In other words, $r \act g = 0$ for all $g$
if and only if $r = 0$.
For any finite ring $R$ and positive integer $k$,
the $M_k(R)$-module $R^k$ is faithful,
so vector and scalar linear codes over rings
are special cases of linear codes over faithful modules.

For a fixed ring $R$, there are generally multiple modules over $R$.
For example, if $R$ is a subring of $S$,
then $(S,+)$ is an $R$-module where the action is multiplication in $S$,
and $(R,+)$ is also an $R$-module where the action is multiplication in $R$.
The following lemma shows that
the linear solvability of a network over a faithful $R$-module is determined entirely by the ring of scalars $R$
and not by the module's underlying Abelian group.
However, we note that not every ring and group pair can form a module.
For example, the additive group of $\GF{2}$
cannot be a $\GF{3}$-module,
since $1 + 1 = 0$ in $\GF{2}$ and $1 + 1 \ne 0$ in $\GF{3}$.

\begin{lemma}
  Let $R$ be a fixed ring.
  If a network is linearly solvable over some faithful $R$-module,
  then it is linearly solvable over every $R$-module.
  \label{lem:same_ring}
\end{lemma}
\begin{proof}
  Let $\Network$ be a network that is linearly solvable over 
  the faithful $R$-module $(G,\oplus)$,
  and let $z_1,\dots,z_m \in G$ denote the messages of $\Network$.
  Suppose a node in $\Network$ has inputs 
  $x_1,\dots,x_n \in G$ in a solution over $\Module{G}{R}$, 
  where, for each $i = 1,\dots,n$,
  $$x_i = (A_{i,1} \act z_1) \oplus \cdots \oplus (A_{i,m} \act z_m)$$
  for some $A_{i,1},\dots,A_{i,m} \in R$.
  Then for each out-edge of this node,
  there exist constants $B_1,\dots,B_n \in R$ such that
  the edge carries the symbol
  \begin{align*}
    \bigoplus_{i=1}^{n} (B_{i} \act x_i)
      &= \bigoplus_{i=1}^{n} \bigoplus_{j=1}^{m} ((B_{i} A_{i,j}) \act z_j)
      = \bigoplus_{j=1}^{m} \left(\left( \sum_{i = 1}^{n} B_i A_{i,j} \right) \act z_j \right).
  \end{align*}
  Then, by induction,
  every edge and decoding function in a linear code over a module 
  is a linear combination of the network messages.
  
  $G$ is a faithful $R$-module, 
  so $1$ and $0$ are the only elements of $R$ such that $1 \act g = g$ and $0 \act g = 0$ for all $g \in G$.
  Hence it must be the case that decoding functions in the linear solution over $\Module{G}{R}$ 
  are of the form
  $$(1 \act z_i) \oplus \bigoplus_{\substack{j = 1 \\ j \ne i}}^n (0 \act z_j).$$
  
  If $H$ is some other $R$-module,
  then a linear solution for $\Network$ over $\Module{G}{R}$
  is also a linear solution for $\Network$ over $\Module{H}{R}$,
  since every edge will carry the same linear combination of the messages 
  (i.e. the same elements of $R$ are the coefficients in the linear combination),
  so, in particular, the decoding functions will be the same linear combination of the messages.
\end{proof}

In contrast to Lemma~\ref{lem:same_ring}, 
if $G$ is both an $R$-module and an $S$-module,
then there may exist a network that is linearly solvable over $\Module{G}{S}$
but not $\Module{G}{R}$.
For example, $\GF{2}$ is a subfield of $\GF{4}$,
so $(\GF{4},+)$ is both a faithful $\GF{2}$-module and a faithful $\GF{4}$-module.
We demonstrate (in Corollary~\ref{cor:choose-two-ring})
a network that is scalar linearly solvable over $\GF{4}$
but not $\GF{2}$,
and by Lemma~\ref{lem:same_ring},
this network is linearly solvable over the $\GF{4}$-module $(\GF{4},+)$
but not the $\GF{2}$-module $(\GF{4},+)$.

The following corollary is a special case of Lemma~\ref{lem:same_ring}
and will be frequently used in later proofs.
It demonstrates an equivalence between scalar linear solutions over matrix rings
and vector linear solutions over rings.
\begin{corollary}
  Let $R$ be a finite ring, $k$ a positive integer, and $\Network$ a network.
  Then $\Network$ is scalar linearly solvable over the ring of $k \times k$ matrices whose elements are from $R$
  if and only if $\Network$ has a $k$-dimensional vector linear solution over $R$.
  \label{cor:same_ring}
\end{corollary}
\begin{proof}
  The ``if'' and the ``only if'' directions are each obtained by separately applying Lemma~\ref{lem:same_ring},
  since $M_k(R)$ and $R^k$ are faithful $M_k(R)$-modules with matrix-matrix multiplication
  and matrix-vector multiplication, respectively.
\end{proof}

Note that in a $k$-dimensional vector linear code over a ring $R$,
the alphabet size is $|R|^k$, whereas in a scalar linear solution over $M_k(R)$,
the alphabet size is $|R|^{k^2}$. 
So any network that is scalar linearly solvable over the matrix ring $M_k(R)$
is also linearly solvable over a smaller module alphabet.
We will generalize this idea in Theorem~\ref{thm:min_module}.

As is common in mathematics literature,
it will be assumed throughout this paper that ring homomorphisms preserve both additive and multiplicative identities.

\begin{lemma}
  If $\phi: R \to S$ is a ring homomorphism 
  and network $\Network$ is linearly solvable over some faithful $R$-module,
  then $\Network$ is linearly solvable over every $S$-module.
  \label{lem:ModHomomorphism}
\end{lemma}
\begin{proof}
  Let $H$ be an $S$-module and
  define a mapping $\newaction: R \times H \to H$ by
  $r \newaction h = \phi(r) \act h$,
  where $\act$ is the action of $\Module{H}{S}$.
  One can verify 
  that $H$ is an $R$-module under $\newaction$.
  Now, let $G$ be a faithful $R$-module,
  and suppose
  $\Network$ has a linear solution over $\Module{G}{R}$.
  By Lemma~\ref{lem:same_ring},
  $\Network$ is linearly solvable over $\Module{H}{R}$,
  so every edge function in the solution over $\Module{H}{R}$ is of the form
  \begin{align}
    y' &= (M_1 \newaction x_1) \oplus \cdots \oplus (M_m \newaction x_m)
  \label{eq:100}
  \end{align}
  where $x_1,\dots,x_m \in H$ are the parent node's inputs 
  and $M_1,\dots,M_m \in R$ are constants.

  Form a linear code for $\Network$ over $\Module{H}{S}$ 
  by replacing each coefficient $M_i$ in \eqref{eq:100} by $\phi(M_i)$.
  Let $y$ be the edge symbol in the code over $\Module{H}{S}$
  corresponding to $y'$ in the code over $\Module{H}{R}$.
  Then
  \begin{align*}
  y &= (\phi(M_1) \act x_1) \oplus \cdots \oplus (\phi(M_m) \act x_m)  \\
    &= (M_1 \newaction x_1) \oplus \cdots \oplus (M_m \newaction x_m) = y'.
  \end{align*}
  Thus, whenever an edge function in the solution over $\Module{H}{R}$ outputs the symbol $y'$,
  the corresponding edge function in the code over $\Module{H}{S}$ will output the same symbol $y'$.
  Likewise, whenever $x$ is an input to an edge function in the solution over $\Module{H}{R}$,
  the corresponding input of the corresponding edge function in the code over $\Module{H}{S}$ 
  will be the same symbol $x$.
  The same argument holds for the decoding functions in the code over $\Module{H}{S}$, so each
  receiver will correctly obtain its corresponding demands in the code over $\Module{H}{S}$.
  Hence, the code over $\Module{H}{S}$ is a linear solution for $\Network$.
\end{proof}

\begin{corollary}
  Let $R$ and $S$ be finite rings.
  If there exists a ring homomorphism from $R$ to $S$,
  then every network that is scalar linearly solvable over $R$
  is also scalar linearly solvable over $S$.
  \label{cor:ModHomomorphism}
\end{corollary}
\begin{proof}
  $(R,+)$ is a faithful $R$-module for any finite ring $R$,
  so this is a special case of Lemma~\ref{lem:ModHomomorphism} 
  where the modules are $\Module{R}{R}$ and $\Module{S}{S}$.
\end{proof}

For finite rings $R$ and $S$,
special cases of Corollary~\ref{cor:ModHomomorphism} include:
\begin{itemize}
\item[(1)] $S$ is a subring of $R$:\\
The identity mapping is an injective homomorphism from $S$ to $R$,
so any network that is scalar linearly solvable over $S$
is also scalar linearly solvable over $R$.
\item[(2)] $R$ has a two-sided ideal $I$:\\
There is a surjective homomorphism from $R$ to $R/I$
(see Lemma~\ref{lem:surjective_homomorphism}),
so any network that is scalar linearly solvable over $R$
is also scalar linearly solvable over $R/I$.
\item[(3)] $\phi: R \times S \to R$ is the projection mapping:\\
$\phi$ is a surjective homomorphism, so
any network that is scalar linearly solvable over $R \times S$
is also scalar linearly solvable over $R$ 
(and likewise over $S$).
\end{itemize}

Cases (1), (2), and (3) agree with 
Lemma \PartOneLemmaSubring,
Corollary \PartOneCorollaryIdeal,
and
Lemma \PartOneLemmaDirectProduct,
respectively,
from \PartOne.
In fact, Corollary~\ref{cor:ModHomomorphism}
is a generalization of these results.

\clearpage
\section{Commutative and non-commutative rings}\label{ssec:simple}

We will focus on linear codes over modules whose ring acts on its own Abelian group,
i.e. scalar linear codes over rings.
As noted after Corollary~\ref{cor:ModHomomorphism},
for any two-sided ideal $I$ of a finite ring $R$,
every network that is scalar linearly solvable over $R$ is also scalar linearly solvable over $R/I$,
so in determining the smallest ring over which a network is scalar linearly solvable,
it is natural to focus attention on rings without two-sided ideals.

A ring is \textit{simple} if it has no proper two-sided ideals.
That is, its only two-sided ideals are the ring itself and the trivial ideal $\{0\}$.
The following lemmas give results related to simple rings
and network linear solvability.

\begin{lemma}
  A finite ring is simple if and only if
  it is isomorphic to a matrix ring over a field.
  \label{lem:simple_rings}
\end{lemma}
\begin{proof}
  This is a corollary of the Artin-Wedderburn theorem 
  (e.g. \cite[p. 36, Theorem 3.10 (4)]{Lam-Noncommutative} and \cite[p. 20, Theorem II.9]{McDonald-FiniteRings}).
\end{proof}

\begin{lemma}{\cite[Theorem 7, p. 243]{Dummit-Algebra}}
  If $I$ is a two-sided ideal of ring $R$, then
  the mapping $\phi: R \to R/I$
  given by
  $\phi(x) = x + I$
  is a surjective homomorphism.
  \label{lem:surjective_homomorphism}
\end{lemma}

\begin{lemma}
  For each finite ring $R$,
  there exists a simple ring $S$ 
  such that the following hold:
  \begin{itemize}
    \item [(a)] there exists a surjective homomorphism from $R$ to $S$,
    \item [(b)] every network that is scalar linearly solvable over $R$
    is scalar linearly solvable over $S$, and
    \item [(c)] $|S|$ divides $|R|$.
  \end{itemize}
  \label{lem:hom_to_simple}
\end{lemma}
\begin{proof}
  If $R$ is a simple ring,
  then each statement is trivially true by taking $S = R$,
  so we may assume $R$ is not a simple ring.
  %
  Thus, $R$ has a proper maximal two-sided ideal $I$.
  Let $S = R/I$, and note that
  since $I$ is maximal, $S$ is simple.
  The mapping $\phi: R \to R/I$
  given by
  $\phi(x) = x + I$
  is a surjective homomorphism
  by Lemma~\ref{lem:surjective_homomorphism},
  which proves (a).
  Hence by Corollary~\ref{cor:ModHomomorphism},
  any network that is scalar linearly solvable over $R$
  is also scalar linearly solvable over $S$,
  which proves (b).
  Since $R$ is finite, we know that
  $|R/I|$ divides $|R|$,
  which proves (c).
\end{proof}

If $R$ is a finite commutative ring and $S$ is a simple ring satisfying (a)-(c) in Lemma~\ref{lem:hom_to_simple},
then $S$ must also be commutative, 
since there is a surjective homomorphism from $R$ to $S$.
However, as we demonstrate in the following example, 
if $R$ is non-commutative, then such an $S$ is not necessarily non-commutative.

\begin{example}
The following demonstrates:
(i) a class of non-commutative rings for which the simple ring in Lemma~\ref{lem:hom_to_simple} is non-commutative,
and 
(ii) a class of non-commutative rings for which the simple ring in Lemma~\ref{lem:hom_to_simple} is commutative
  \begin{itemize}
    \item[(i)] Let $\Z_n$ denote the ring of integers mod $n$.
    For any positive integers $k,n$, and prime divisor $p$ of $n$,
    there exists a surjective homomorphism from
    the non-commutative ring $M_k(\Z_n)$ to
    the non-commutative simple ring $M_k(\Z_p)$,
    given by matrix-component-wise reduction mod $p$.
    
    \item[(ii)] For each field $\F$ and integer $k \geq 2$,
    there exists a surjective homomorphism from
    the non-commutative ring of upper triangular $k \times k$ matrices with entries in $\F$
    to the commutative simple ring $\F$
    (see the proof of Lemma~\ref{lem:p3_matrices}).
  \end{itemize}
\end{example}

The following theorem demonstrates that any smallest ring over which a network
is scalar linearly solvable is simple.
\begin{theorem}
  If a network is scalar linearly solvable over a ring $R$
  but not over any smaller ring,
  then $R$ is a matrix ring over a field.
  \label{thm:R_dom_by_simple}
\end{theorem}
\begin{proof}
  Suppose a network $\Network$ is scalar linearly solvable over a
  ring $R$ that is not simple.
  %
  By Lemma~\ref{lem:hom_to_simple} (a) (b),
  there exists a simple ring $S$ and
  a surjective homomorphism $\phi: R \to S$,
  such that
  $\Network$ is scalar linearly solvable over $S$.
  Since $\phi$ is surjective, $|R| \geq |S|$,
  but since $S$ is simple and $R$ is not, the two rings cannot be isomorphic,
  so $|R| \ne |S|$, and therefore $|R| > |S|$.

  This proves that every smallest size ring over which $\Network$
  is scalar linearly solvable must be simple,
  which implies that such a ring is a matrix ring over a field
  by Lemma~\ref{lem:simple_rings}.
\end{proof}

In \PartOne{} \cite[Theorem \PartOneTheoremSmallerField]{Connelly-Ring1}, 
we showed that the smallest-size commutative ring over which a network is scalar linearly solvable is unique.
However, there may exist multiple simple rings of the same size 
(e.g. $\GF{p^4}$ and $M_2(\GF{p})$ are non-isomorphic simple rings of size $p^4$).
An interesting open question is whether every network with a scalar linear solution over multiple simple rings of the same size
also must have a scalar linear solution over some smaller simple ring.
I.e. is the smallest ring $R$ in Theorem~\ref{thm:R_dom_by_simple} unique for a given network?

We demonstrate (in Corollaries~\ref{cor:choose-two-ring} and \ref{cor:dim-n_smallest_ring})
that for two infinite classes of networks (one of which is a class of multicast networks)
studied in this paper, 
the smallest-size ring over which each network is scalar linearly solvable is unique.

\subsection{Modules and vector linear codes}\label{ssec:mod_vec}

The following lemma shows that linear solutions over unfaithful modules
admit linear solutions over faithful modules.
\begin{lemma}
  Let $G$ be an $R$-module.
  There exists a ring $S$ such that
  $G$ is a faithful $S$-module,
  and any network that is linearly solvable over $\Module{G}{R}$
  is linearly solvable over $\Module{G}{S}$.
  If $R$ is commutative,
  then there exists a commutative such $S$.
  \label{lem:faithful_module}
\end{lemma}
\begin{proof}
  We use ideas from \cite[p. 2750]{DFZ-Insufficiency} here.
  Let $J = \{r \in R \; : \; r \act g = 0, \; \forall g \in G\}$,
  which is easily verified to be a two-sided ideal of $R$.
  Let $S = R/J$.
  It can also be verified that
  $G$ is an $S$-module with action $\newaction$ given by
  $(r + J) \newaction g = r \act g$.
  
  If $(r + J), (s + J) \in S$ are such that
  $(r + J) \newaction g = (s + J) \newaction g$
  for all $g \in G$,
  then $(r-s) \act g = 0$,
  which implies $(r - s) \in J$.
  Hence $(r + J) = (s + J)$,
  so the ring $S$ acts faithfully on $G$.
  If $R$ is commutative, then the ring $R/J = S$ is also commutative.

  Suppose a network $\Network$ is linearly solvable over $\Module{G}{R}$.
  Every edge function in the solution is of the form
  \begin{align}
  y' = M_1 \act x_1 + \cdots + M_m \act x_m
  \label{eq:t}
  \end{align}
  where the $x_i$'s are the parent node's inputs and the $M_i$'s are constants from $R$.
  Form a linear code over $\Module{G}{S}$ replacing each coefficient $M_i$ in \eqref{eq:t}
  by $(M_i + J)$.
  Let $y$ be the edge symbol in the code over $\Module{G}{S}$
  corresponding to $y'$ in the code over $\Module{G}{R}$.
  Then
  \begin{align*}
  y &= ((M_1 + J) \newaction x_1) \oplus \cdots \oplus ((M_m + J) \newaction x_m)  \\
    &= (M_1 \act x_1) \oplus \cdots \oplus (M_m \act x_m) = y'.
  \end{align*}
  %
  %
  Thus, whenever an edge function in the solution over $\Module{G}{R}$ outputs the symbol $y'$,
  the corresponding edge function in the code over $\Module{G}{S}$ will output the same symbol $y'$.
  Likewise, whenever $x$ is an input to an edge function in the solution over $\Module{G}{R}$,
  the corresponding input of the corresponding edge function in the code over $\Module{G}{S}$ 
  will be the same symbol $x$.
  The same argument holds for the decoding functions in the code over $\Module{G}{S}$, so each
  receiver will correctly obtain its corresponding demands in the code over $\Module{G}{S}$.
  Hence, the code over $\Module{G}{S}$ is a linear solution for $\Network$.
\end{proof}

In a linear network code over a module $\Module{G}{R}$,
in principle, the ring $R$ need not be finite
(although representing linear code coefficients might be problematic).
For example, any Abelian group $(G,\oplus)$
is a $\Z$-module with action given by 
$$n \act g =
  \left\{ \begin{array}{ll}
    \underbrace{g \oplus \cdots \oplus g}_{n \text{ adds}} & n > 0 \\    
    (-n) \act (-g) & n < 0 \\
    0 & n = 0 .
  \end{array} \right. $$
However, in a linear network code over a module, the alphabet is finite,
so the Abelian group $G$ must be finite.%
\footnote{We will call a module ``finite'' if and only if its Abelian group is finite.}
The following corollary shows that if a network is linearly solvable over a module
where the ring is infinite,
then it is also linearly solvable over a faithful module where the ring is finite.

\begin{corollary}
  Let $R$ be an infinite ring
  and let $G$ be a finite $R$-module.
  Then there exists a finite ring $S$
  such that $G$ is a faithful $S$-module
  and any network that is linearly solvable over $\Module{G}{R}$
  is linearly solvable over $\Module{G}{S}$.
  If $R$ is commutative,
  then there exists a commutative such $S$.
  \label{cor:infinite_ring}
\end{corollary}
\begin{proof}
  This follows from Lemma~\ref{lem:faithful_module},
  and the fact that
  the ring of a faithful finite module
  must also be finite.
\end{proof}

A \textit{submodule} of an $R$-module $G$ is a subgroup $H$ of $G$
such that $H$ is closed when acted on by $R$. 
That is, both $H$ and $G$ are $R$-modules
and $H \subseteq G$.
Submodules are of particular interest, since 
by Lemma~\ref{lem:same_ring},
if $G$ and $H$ are faithful $R$-modules,
then the set of networks that are linearly solvable over $\Module{G}{R}$
and the set of networks that are linearly solvable over $\Module{H}{R}$
are equal, yet a linear code over $\Module{H}{R}$ has a smaller alphabet
if $H$ is a proper submodule of $G$.

As an example, let $I$ be a two-sided ideal in the ring $R$.
Then $(I,+)$ is a subgroup of $(R,+)$ that is closed under multiplication in $R$,
so $\Module{I}{R}$ is a submodule of the $R$-module $R$.
As another example,
for each finite field $\F$ and integer $k \geq 2$,
the $M_k(\F)$-module $\F^k$ is a proper submodule of
the $M_k(\F)$-module $M_k(\F)$.

Lemmas~\ref{lem:submodules} and \ref{lem:smaller_module}
show results related to submodules
that will be used to prove Theorem~\ref{thm:min_module}.

\begin{lemma}{\cite[Theorem 3.3 (2), p. 31]{Lam-Noncommutative}}
  Let $\F$ be a finite field and $k$ a positive integer.
  Then $\F^k$ is the only $M_k(\F)$-module that has no proper submodules.
  \label{lem:submodules}
\end{lemma}
%

By Lemma~\ref{lem:same_ring},
for each ring $R$,
if a network is linearly solvable over a faithful $R$-module,
then it is linearly solvable over every $R$-module.
When a network is solvable over the $R$-modules for a particular ring $R$,
it may be desirable for linear network coding to determine
the minimum-size $R$-modules.
Lemma~\ref{lem:smaller_module} considers this question for rings of matrices over a finite field.

\begin{lemma}
  Let $\F$ be a finite field and $k$ a positive integer.
  If $G$ is a finite non-zero $M_k(\F)$-module,
  then $|\F|^k$ divides $|G|$.
  \label{lem:smaller_module}
\end{lemma}
\begin{proof}
  Since $G$ is finite and non-zero,
  $G$ contains a submodule with no proper submodules.
  By Lemma~\ref{lem:submodules}, 
  $\F^k$ is the only $M_k(\F)$-module with no proper submodules,
  so $\F^k$ is a submodule of $G$.
  Hence by Lagrange's theorem of finite groups 
  (e.g. \cite[p. 89, Theorem 8]{Dummit-Algebra}),
  $|\F|^k$ divides $|G|$.
\end{proof}

The following theorem is a generalization of Theorem~\ref{thm:R_dom_by_simple},
where we characterize smallest-size modules over which networks are linearly solvable.
Theorem~\ref{thm:min_module} 
demonstrates that if a network is linearly solvable over some module,
then there exists a vector linear code over a field that
minimizes the alphabet size needed for a linear solution.

\begin{theorem}
  Suppose a network $\Network$ is linearly solvable over an $R$-module $G$. 
  Then the following hold:
  \begin{itemize}
  \item[(a)] There exists a finite field $\F$ and positive integer $k$
  such that $\Network$ has a $k$-dimensional vector linear solution over $\F$
  and $|\F|^k$ divides $|G|$.
  \item[(b)] If $R$ is commutative, then there exists a finite field $\F$
  such that $\Network$ has a scalar linear solution over $\F$
  and $|\F|$ divides $|G|$.
  \end{itemize}
  \label{thm:min_module}
\end{theorem}
\begin{proof}
  If the ring $R$ is infinite, then by Corollary~\ref{cor:infinite_ring},
  $\Network$ is linearly solvable over some faithful module with a finite ring.
  If $R$ is commutative, then by Corollary~\ref{cor:infinite_ring},
  $\Network$ is linearly solvable over some faithful module with a finite commutative ring.
  So without loss of generality, assume $R$ is finite
  and $G$ is a faithful $R$-module.
  By Lemmas~\ref{lem:simple_rings} and \ref{lem:hom_to_simple} (a),
  since $R$ is finite,
  there exists a field $\F$, a positive integer $k$,
  and a surjective homomorphism $\phi: R \to M_k(\F)$.
  By Lemma~\ref{lem:ModHomomorphism}
  any network that is linearly solvable over the faithful $R$-module $G$
  is also linearly solvable over every $M_k(\F)$-module,
  so in particular, 
  $\Network$ has a $k$-dimensional vector linear solution over $\F$.
  Since $\phi$ is a homomorphism,
  any $R$-module is also an $M_k(\F)$-module 
  (see the proof of Lemma~\ref{lem:ModHomomorphism}).
  Thus, both $G$ and $\F^k$ are $M_k(\F)$-modules,
  so by Lemma~\ref{lem:smaller_module},
  we have $|\F|^k$ divides $|G|$.

  If $R$ is commutative, then, since $\phi$ is a surjective homomorphism, 
  $M_k(\F)$ must also be commutative,
  which implies $k = 1$.
  Hence $\Network$ has a scalar linear solution over $\F$
  and $|\F|$ divides $|G|$.
\end{proof}

Theorem~\ref{thm:min_module} demonstrates that, in some sense,
vector linear codes over finite fields are optimal
for linear network coding,
as they can minimize the alphabet size needed for a linear solution.
The following lemmas will be used to show (in Theorem~\ref{thm:min_module_not_unique}) 
that a minimum-size module
over which a network is linearly solvable is not necessarily unique.
Lemma~\ref{lem:kn-vector} is a result of Sun et. al \cite{Sun-VL}.

\begin{lemma}{\cite[Proposition 1, p. 4513]{Sun-VL}}
  Let $q$ be a prime power and $k$ a positive integer.
  If a network has a scalar linear solution over $\GF{q^k}$,
  then it has a $k$-dimensional vector linear solution over $\GF{q}$.
  \label{lem:kn-vector}
\end{lemma}

For each integer $n \geq 3$,
the \textit{$n$-Choose-Two Network} is a multicast network that was described by
Rasala Lehman and Lehman \cite{Lehman-Complexity}
and further studied in our \PartOne{}
(see Figure~\PartOneNChooseTwoFigure{} in \cite{Connelly-Ring1}).
%
\begin{lemma}{\cite[p. 144]{Lehman-Complexity}}  
  Let $\A$ be a network alphabet and let integer $n \ge 3$. 
\begin{itemize}
    \item[(a)] If the $n$-Choose-Two Network has a solution over $\A$,
      then $|\A| \geq n - 1$.
    \item[(b)]
     Let $\A$ be a field. 
     The $n$-Choose-Two Network is linearly solvable over $\A$
      if and only if
      $|\A| \ge n - 1$.
  \end{itemize}
\label{lem:kchoosetwo_fields}
\end{lemma}

\begin{theorem}
  For each integer $k \geq 2$ and prime $p$,
  the $(p^k+1)$-Choose-Two Network is linearly solvable over at least two distinct modules 
  of size $p^k$
  but not over over any smaller modules.
  \label{thm:min_module_not_unique}
\end{theorem}
\begin{proof}
  By Lemma~\ref{lem:kchoosetwo_fields},
  the $(p^k+1)$-Choose-Two Network
  is scalar linearly solvable over $\GF{p^k}$
  and is not solvable over any alphabet whose size is less than $p^k$.
  By Lemma~\ref{lem:kn-vector},
  any network with a scalar linear solution over $\GF{p^k}$
  has a $k$-dimensional vector linear solution over $\GF{p}$.
  Hence the $(p^k+1)$-Choose-Two Network has a scalar linear solution over $\GF{p^k}$
  and a $k$-dimensional vector linear solution over $\GF{p}$,
  yet the network has no linear solution over any module whose size is less than $p^k$.
\end{proof}

The following corollary generalizes Theorem \PartOneTheoremFieldBeatsRing{} from \PartOne,
which showed the $(p^k+1)$-Choose-Two Network is not scalar linearly solvable
over any commutative ring of size $p^k$ other than the field $\GF{p^k}$.
In fact, as a result of Corollary~\ref{cor:choose-two-ring}, the $(p^k+1)$-Choose-Two Network
is not scalar linearly solvable over any ring of size $p^k$ other than the field.
\begin{corollary}
  For each integer $k \geq 2$ and prime $p$,
  the unique smallest-size ring over which the $(p^k+1)$-Choose-Two Network is scalar linearly solvable 
  is $\GF{p^k}$.
  \label{cor:choose-two-ring}
\end{corollary}
\begin{proof}
  By Lemma~\ref{lem:kchoosetwo_fields},
  the $(p^k+1)$-Choose-Two Network
  is scalar linearly solvable over $\GF{p^k}$
  and is not solvable over any smaller alphabet.

  Suppose the $(p^k+1)$-Choose-Two Network is scalar linearly solvable over a ring $R$
  of size $p^k$.
  By Lemmas~\ref{lem:simple_rings} and \ref{lem:hom_to_simple} (a) (b),
  there exists a field $\F$, a positive integer $n$,
  and a surjective homomorphism $\phi:R \to M_n(\F)$ such that
  the $(p^k+1)$-Choose-Two Network is scalar linearly solvable over the ring $M_n(\F)$.
  Since $\phi$ is surjective, 
  $p^k \geq |\F|^{n^2}$.
  By Corollary~\ref{cor:same_ring},
  the $(p^k + 1)$-Choose-Two Network has an $n$-dimensional vector linear solution over $\F$,
  so by Lemma~\ref{lem:kchoosetwo_fields} (a),
  $|\F|^n \geq p^k$.

  Hence $|\F|^{n} \geq p^k \geq |\F|^{n^2}$,
  which implies $n = 1$ and $\F = \GF{p^k}$.
  Since $\phi:R \to \F$ is a surjective homomorphism
  and $|\F| = |R|$,
  we have $R \cong \GF{p^k}$.
\end{proof}

The following corollaries summarize our results on the linear solvability of networks
using scalar and linear vector codes over fields,
scalar linear codes over rings,
and linear codes over modules.
Corollary~\ref{cor:general_vector} shows an equivalence between vector linear solvability over fields
and linear solvability over rings and modules,
while Corollary~\ref{cor:general_scalar} shows an equivalence between scalar linear solvability over fields
and linear solvability over commutative rings and modules.

\begin{corollary}
  For any network $\Network$,
  the following three statements are equivalent:
  \begin{itemize}
    \item[(i)] $\Network$ is vector linearly solvable over some finite field.
    \item[(ii)] $\Network$ is scalar linearly solvable over some ring.
    \item[(iii)] $\Network$ is linearly solvable over some module.
  \end{itemize}
  \label{cor:general_vector} 
\end{corollary}
\begin{proof}
  If a network has a $k$-dimensional vector linear solution over some field $\F$,
  then by Corollary~\ref{cor:same_ring} it has a scalar linear solution over the ring $M_k(\F)$,
  hence (i) implies (ii).
  A scalar linear code over a ring is a special case of a linear code over a module,
  so (ii) implies (iii).
  By Theorem~\ref{thm:min_module} (a),
  (iii) implies (i).
\end{proof}

%
\begin{corollary}
  For any network $\Network$,
  the following three statements are equivalent:
  \begin{itemize}
    \item[(i)] $\Network$ is scalar linearly solvable over some finite field.
    \item[(ii)] $\Network$ is scalar linearly solvable over some commutative ring.
    \item[(iii)] $\Network$ is linearly solvable over some module whose ring is commutative.
    \label{cor:general_scalar}
  \end{itemize}
\end{corollary}
\begin{proof}
  A scalar linear code over a finite field is a special case of a scalar linear code over a commutative ring,
  hence (i) implies (ii).
  A scalar linear code over a commutative ring is a special case of a linear code over a module
  where the ring is commutative,
  so (ii) implies (iii).
  By Theorem~\ref{thm:min_module} (b),
  (iii) implies (i).
\end{proof}

We summarize our results on minimizing the alphabet size in linear network coding by:
\begin{itemize}
  \item If a network is scalar linearly solvable over some commutative ring,
  then the (unique) smallest such commutative ring is a field \cite[Theorem \PartOneTheoremSmallerField]{Connelly-Ring1}.

  \item If a network is scalar linearly solvable over some ring,
  then a smallest such ring is a matrix ring over field (Theorem~\ref{thm:R_dom_by_simple}).
  It is not known whether such a smallest ring is unique.
  
  \item If a network is linearly solvable over some module,
  then a smallest such module yields
  a vector linear solution over a field (Theorem~\ref{thm:min_module}).
  Such a module may not be unique (Theorem~\ref{thm:min_module_not_unique}).
\end{itemize}

\clearpage
\section{The Dim-$n$ Network} \label{sec:Gen_M}

\psfrag{x1}{$\xx{1}{1},\dots,\xx{1}{n}$}
\psfrag{xn}{$\xx{n}{1},\dots,\xx{n}{n}$}
\psfrag{e1}{$\ww{1}{1},\dots,\ww{1}{n-1}$}
\psfrag{en}{$\ww{n}{1},\dots,\ww{n}{n-1}$}
\psfrag{f1}{$\ww{1}{n}$}
\psfrag{fn}{$\ww{n}{n}$}
\psfrag{g1}{$u_1$}
\psfrag{gnn}{$u_{n^n}$}
\psfrag{a1}{$a_1$}
\psfrag{an}{$a_{n}$}
\psfrag{b1}{$b_1$}
\psfrag{bn}{$b_{n}$}
\psfrag{Z}{$Z$}
\psfrag{R1}{$R_1$}
\psfrag{Rnn}{$R_{n^n}$}
\psfrag{S1}{$\xx{1}{1},\dots,\xx{n}{1}$}
\psfrag{Snn}{$\xx{1}{n},\dots,\xx{n}{n}$}
\psfrag{e11}{$\ww{1}{1}$}
\psfrag{e1n}{$\ww{1}{n-1}$}
\psfrag{en1}{$\ww{n}{1}$}
\psfrag{enn}{$\ww{n}{n-1}$}
\psfrag{n-1}{$n-1$}
\begin{figure}[h]
  \begin{center}
    \leavevmode
    \hbox{\epsfxsize=.65\textwidth\epsffile{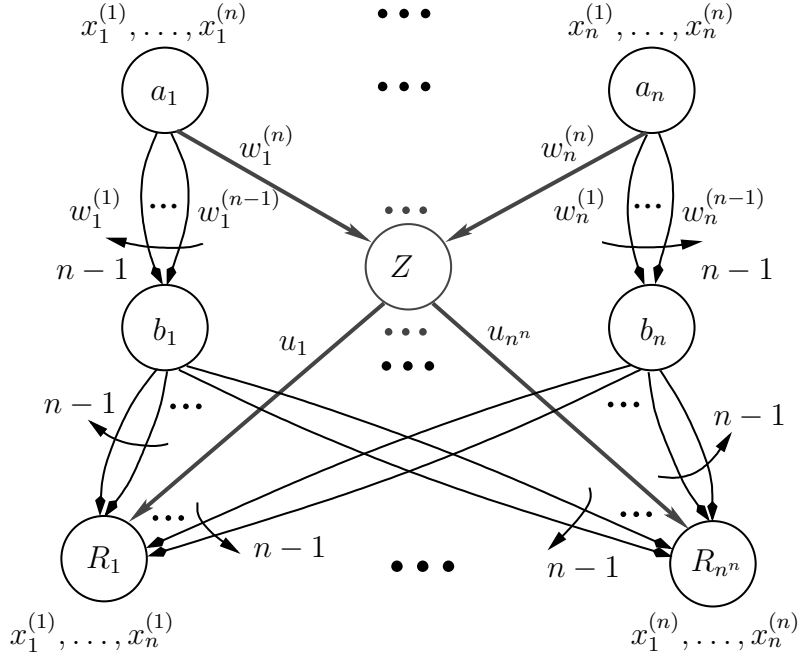}}
  \end{center}
  \caption{The Dim-$n$ network.
  For each $i = 1,\dots,n$, 
  the node $a_i$ is a source node that generates messages $\xx{i}{1},\dots,\xx{i}{n}$,
  and $a_i$ has $n-1$ parallel out-edges to node $b_i$
  and one out-edge to node $Z$.
  For each $j = 1,\dots,n^n$,
  the receiver $R_j$ has $n-1$ parallel in-edges from each of the nodes $b_1,\dots,b_n$
  and a single in-edge from node $Z$.
  Each receiver demands a single message from each source node
  and each set of $n$ messages demanded by each receiver is unique;
  that is, for any  $i_1, \dots, i_n \in \{1, \dots, n\}$,
  there is exactly one receiver which demands $\xx{1}{i_1},\dots,\xx{n}{i_n}$.
  }
\label{fig:Dim-n}
\end{figure}

For each integer $n \geq 2$,
the \textit{Dim-$n$ Network} is defined in Figure~\ref{fig:Dim-n}
and is referred to as such because
it has vector linear solutions precisely 
over vector dimensions that are multiples of $n$.
We prove this fact in Theorem~\ref{thm:dim-n_vector}.
This infinite family of networks will be used to demonstrate
several theorems related to commutative and non-commutative rings.
The special case of $n = 2$ corresponds to the \textit{M Network}
of \cite{Medard-NonMulticast},
shown later in Figure~\ref{fig:M-network}.

\begin{remark}
  The Dim-$n$ Network has $n^n + 2n + 1$ nodes
  and $n^n ( n^2 - n + 1) + n^2$ edges.
  %
  \label{rem:dim-n_size}
\end{remark}

\begin{lemma}
  For each integer $n \geq 2$ and alphabet $\A$,
  the Dim-$n$ Network has an $n$-dimensional vector routing solution over $\A$.
  \label{lem:dim-n_rout}
\end{lemma}
\begin{proof}
  Each message and edge symbol is an element of $\A^n$.
  Let $[x]_i$ denote the $i$th component of $x \in \A^n$.
  Define an $n$-dimensional routing code over $\A$ by
  \begin{align*}
    \left[ \ww{i}{j} \right]_k &= \left[ \xx{i}{k} \right]_j 
          & & (i,j,k = 1,\dots,n).
  \end{align*}
  That is, the $k$th component of the $j$th out-edge of the $i$th source node
  carries the $j$th component of the $k$th message originating at the $i$th source node.

  For each $i = 1,\dots,n$ and each $j = 1,\dots,n^n$,
  let the set of $(n-1)$ parallel edges from node $b_i$ to receiver $R_j$ carry the symbols
  $\ww{i}{1},\dots,\ww{i}{n-1}$.
  Then each receiver gets the first $(n-1)$ components of every message from the edges originating at $b_1,\dots,b_n$,
  so in particular,
  each receiver can recover the first $(n-1)$ components of each of the messages it demands.

  Node $Z$ receives the $n$th component of each message,
  so each of its out-edges can carry any $n$ of these components.
  Let $j  \in \{1,\dots,n^n \}$,
  suppose $\xx{1}{i_1},\dots,\xx{n}{i_n}$ are the messages receiver $R_j$ demands,
  and let
  \begin{align*}
    \left[ u_j \right ]_k 
      &= \left[ \ww{k}{n}   \right]_{i_k}
      = \left[ \xx{k}{i_k} \right]_n 
        & & (k = 1,\dots,n).
  \end{align*}
  Then $R_j$ can recover the $n$th component of each of the messages it demands.
  Since $j$ was chosen arbitrarily,
  the code is an $n$-dimensional vector routing solution.
\end{proof}

The following lemmas will be used in later proofs.
\begin{lemma}
  Let $R$ be a finite ring and let $k_1,\dots,k_n$ be positive integers.
  If a network has $k_1,\dots,k_n$-dimensional vector linear solutions over $R$,
  then the network has a $(k_1+\cdots+k_n)$-dimensional vector linear solution over $R$.
  \label{lem:dim_sum}
\end{lemma}
\begin{proof}
  Assume a network has a $k_i$-dimensional vector linear solution over $R$ for each $i = 1,\dots,n$.
  In the $k_i$-dimensional vector linear solution over $R$,
  every edge function is of the form
  $$
    y^{(i)} = M_1^{(i)} \xx{1}{i} + \cdots + M_m^{(i)} \xx{m}{i}
  $$
  where $\xx{j}{i} \in R^{k_i}$ are the inputs to the node
  and $M_j^{(i)}$ are $k_i \times k_i$ matrices over $R$.
  For any such edge function,
  define a $(k_1+\cdots+k_n)$-dimensional vector linear edge function over $R$ by letting 
  \begin{align*}
     \left[ \begin{array}{c}
        y^{(1)} \\
        \vdots \\
        y^{(n)}
      \end{array}\right] 
    &= \sum_{j = 1}^m 
      \left[ \begin{array}{ccl}
        M_j^{(1)} & & \text{\LARGE0}  \\ 
        & \ddots  \\
        \text{\LARGE0} & & M_j^{(n)}
      \end{array} \right]
      \,
      \left[ \begin{array}{c}
        \xx{j}{1} \\
        \vdots \\
        \xx{j}{n}
      \end{array}\right] .
  \end{align*}
It is straightforward to see this provides a vector linear solution for the network.
\end{proof}

Let $X$ and $Y$ be collections of discrete random variables
over alphabet $\A$, and let $p_X$ be the probability mass function of $X$.
We denote the (base $|\A|$) \textit{entropy} of $X$ as
$$H(X) = - \sum_{u} p_X(u) \log_{|\A|} \, p_X(u)$$
and the \textit{conditional entropy} of $X$ given $Y$ as
$$H(X | Y) = H(X,Y) - H(Y).$$
The proof of Theorem~\ref{thm:dim-n_vector} will make use of
Lemmas~\ref{lem:entropy_sum} and \ref{lem:int_ent} and the following basic information inequalities:
\begin{align}
  H(X|Y) & \leq H(X) 
    \label{eq:ent_1}\\
         & \leq H(X,Y) 
    \label{eq:ent_2}\\
         & \le H(X) + H(Y).
    \label{eq:ent_3}
\end{align}

\begin{lemma}
  Let $X,Y_1,\dots,Y_n$ be collections of discrete random variables.
  Then
  $$\entropy{X,Y_1} + \cdots + \entropy{X,Y_n} \geq (n-1) \entropy{X} + \entropy{X,Y_1,\dots,Y_n}.$$%
  \label{lem:entropy_sum}
\end{lemma}
\begin{proof}
  \begin{align*}
    \sum_{i=1}^n \entropy{X,Y_i} &= n \entropy{X} + \sum_{i=1}^n \entropy{Y_i | X} \\
      & \geq n \entropy{X} + \entropy{Y_1 | X} +  \sum_{i=2}^n \entropy{Y_i | X, Y_1,\dots,Y_{i-1}} 
        & & \Comment{\eqref{eq:ent_1}}\\
      & = (n-1) \entropy{X} + \entropy{X,Y_1,\dots,Y_n}.
  \end{align*}
\end{proof}

\begin{lemma}{\cite[Lemma V.9]{DFZ-Matroids}}
  Let $L: \F^{m} \to \F^n$ be a linear map,
  and let $x$ be a uniformly distributed random variable on $\F^m$.
  Then $L(x)$ is uniformly distributed on the range of $L$,
  and the base $|\F|$ entropy of $L(x)$ is $H(L(x)) = \dim{ \range{L(x} } \cdot \log{|\F|}$.
  \label{lem:int_ent}
\end{lemma}

\begin{theorem}
  For each integer $n \geq 2$ and each field $\F$,
  the Dim-$n$ Network has a $k$-dimensional vector linear solution over $\F$
  if and only if
  $\Div{n}{k}$.
  \label{thm:dim-n_vector}
\end{theorem}
\begin{proof}
  Suppose $\Div{n}{k}$. 
  Then $k=nc$ for some integer $c \ge 1$.
  By Lemma~\ref{lem:dim-n_rout}, the Dim-$n$ Network has an $n$-dimensional vector linear solution over $\F$,
  so by taking $k_1=\cdots=k_c = n$ in Lemma~\ref{lem:dim_sum},
  the Dim-$n$ Network has an $nc$-dimensional vector linear solution over $\F$.

  Conversely, suppose the Dim-$n$ Network has a $k$-dimensional vector linear solution over field $\F$.
  Then all messages $\xx{i}{j}$ and edge symbols $\ww{i}{j}$ are $k$-vectors over $\F$.
  By viewing the message components as independent uniform random variables over $\F$
  and considering the entropy using logarithms base $|\F|$, we have
  \begin{align}
    \entropy{ \xx{1}{1},\dots,\xx{1}{n},\dots,\xx{n}{1},\dots,\xx{n}{n} } 
      &= \sum_{i,j=1}^n \entropy{ \xx{i}{j} }. 
      \label{eq:entropy_2}
  \end{align}
  For each $i = 1,\dots,n$, the edge symbols
  $\ww{i}{1},\dots,\ww{i}{n}$ are linear functions of
  $\xx{i}{1},\dots,\xx{i}{n}$,
  so
  \begin{align}
    \entropy{\ww{i}{1} ,\dots,\ww{i}{n}  \, | \, \xx{i}{1},\dots, \xx{i}{n} } & = 0
    & & (i = 1,\dots,n).
    \label{eq:entropy_3}
  \end{align}

  The receiver $R_{1}$
  demands the messages $\xx{1}{1},\dots,\xx{n}{1}$
  and recovers its demands from its in-edges, so
  \begin{align}
    \entropy{\xx{1}{1}, \dots, \xx{n}{1} \, | \, \ww{1}{1} ,\dots,\ww{1}{n-1} ,\dots,\ww{n}{1} ,\dots,\ww{n}{n-1} , u_{1} } = 0.
  \label{eq:entropy_4}
  \end{align}
  For each $i,j \in \{1,\dots,n\}$,
  the edge symbol $\ww{i}{j} $ is a linear function of only $ \xx{i}{1},\dots, \xx{i}{n}$,
  and the network's messages are jointly independent.
  Thus,
  \begin{align*}
    &\sum_{i=1}^n \entropy{\ww{i}{1} ,\dots,\ww{i}{n-1} , x_i^{(1)}} \\
      &\; = \entropy{ \xx{1}{1}, \dots, \xx{n}{1}, \ww{1}{1} ,\dots,\ww{1}{n-1} ,\dots,\ww{n}{1} ,\dots,\ww{n}{n-1}  }
        & & \Comment{independence} \\
      &\; \leq \entropy{ u_{1}, \xx{1}{1}, \dots, \xx{n}{1}, \ww{1}{1} ,\dots,\ww{1}{n-1} ,\dots,\ww{n}{1} ,\dots,\ww{n}{n-1} } 
        & & \Comment{\eqref{eq:ent_2}}\\
      &\; = \entropy{u_{1}, \ww{1}{1} ,\dots,\ww{1}{n-1} ,\dots,\ww{n}{1} ,\dots,\ww{n}{n-1}} 
        & & \Comment{\eqref{eq:entropy_4}}
        \\
      &\; \le \entropy{u_{1}} + \sum_{i = 1}^n \sum_{j = 1}^{n-1} \entropy{\ww{i}{j}}
        & & \Comment{\eqref{eq:ent_3}}\\
      &\; \leq k \, (1 + n (n-1)).
  \end{align*}
  By a similar argument,
  for any $i_1, \dots, i_n \in \{1,\dots,n\}$,
  there exists a receiver which demands the messages
  $\xx{1}{i_1},\dots,\xx{n}{i_n}$, so
  \begin{align}
    \sum_{j=1}^n \entropy{\ww{j}{1} ,\dots,\ww{j}{n-1} , \xx{j}{i_j} } 
      & \leq k \, (n^2 - n + 1).
      \label{eq:entropy_5}
  \end{align}
  Since $\displaystyle\bigcup_{j=1}^n \left\{ \ww{j}{1} ,\dots, \ww{j}{n} \right\}$
  is a cut-set for each receiver,
  we have
  \begin{align}
    \entropy{ \xx{1}{1},\dots, \xx{1}{n}, \dots, \xx{n}{1},\dots,\xx{n}{n} \, | \, \ww{1}{1} ,\dots,\ww{1}{n} ,\dots,\ww{n}{1} ,\dots,\ww{n}{n} }
      &= 0
      \label{eq:entropy_6}.
  \end{align}
  Therefore,
  \begin{align*}
    k n^2 &= \entropy{ \xx{1}{1} ,\dots, \xx{1}{n} ,\dots, \xx{n}{1} ,\dots, \xx{n}{n} }
        & & \Comment{\eqref{eq:entropy_2}}
        \\
      & \leq \entropy{ \xx{1}{1} ,\dots, \xx{1}{n} ,\dots, \xx{n}{1} ,\dots, \xx{n}{n} , \ww{1}{1} ,\dots,\ww{1}{n} ,\dots,\ww{n}{1} ,\dots,\ww{n}{n} } 
        & & \Comment{\eqref{eq:ent_2}}\\
      & = \entropy{\ww{1}{1} ,\dots,\ww{1}{n} ,\dots,\ww{n}{1} ,\dots,\ww{n}{n} } 
        & & \Comment{\eqref{eq:entropy_6}}
        \\
      &\le \sum_{i=1}^n \sum_{j=1}^n \entropy{\ww{i}{j}} & & \Comment{\eqref{eq:ent_3}}\\ 
     & \leq k n^2
  \end{align*}
which implies
\begin{align*}
\sum_{i=1}^n \sum_{j=1}^n \entropy{\ww{i}{j}} = k n^2.
\end{align*}
But, since $\entropy{\ww{i}{j}} \le k$, we get
  \begin{align*}
    \entropy{\ww{i}{j} } &= k & & (i,j=1,\dots,n).
  \end{align*}
  Also, since $\ww{1}{1} ,\dots,\ww{1}{n} ,\dots,\ww{n}{1} ,\dots,\ww{n}{n} $ are independent,
  \begin{align}
      \entropy{\ww{i}{1} ,\dots,\ww{i}{n-1} } = k (n-1) & & (i = 1,\dots, n)
      \label{eq:entropy_7}.
  \end{align}
  For each $i = 1,\dots,n$, we have
  \begin{align}
  & \sum_{j=1}^n \entropy{\ww{i}{1} ,\dots,\ww{i}{n-1} , \xx{i}{j} } 
      \notag \\
    &\;  \geq (n-1) \entropy{\ww{i}{1} ,\dots,\ww{i}{n-1}  } + \entropy{\ww{i}{1} ,\dots,\ww{i}{n-1} ,x_i^{(1)},\dots,x_i^{(n)}} 
        & & \Comment{Lemma~\ref{lem:entropy_sum}}
        \notag \\
    &\; = k (n-1) (n-1) + \entropy{ \xx{i}{1} ,\dots, \xx{i}{n} }
        & & \Comment{\eqref{eq:entropy_3}, \eqref{eq:entropy_7}} 
        \notag \\
    & \;= k (n^2 - n + 1) 
        & & \Comment{\eqref{eq:entropy_2}}
        \label{eq:entropy_9}. 
  \end{align}

  By fixing $i_1 = 1$ and summing over all $i_2,\dots,i_n$ in \eqref{eq:entropy_5}, we have
  \begin{align*}
    & n^{n-1} \, k \, (n^2 - n + 1)  \\
        & \; \geq 
          \sum_{i_2,\dots,i_n = 1}^n \left( \entropy{\ww{1}{1} ,\dots,\ww{1}{n-1} , \xx{1}{1} } 
            + \sum_{j=2}^n \entropy{\ww{j}{1} ,\dots,\ww{j}{n-1} , \xx{j}{i_j} }  \right) 
            & & \Comment{\eqref{eq:entropy_5}}\\
        &\; = n^{n-1} \entropy{\ww{1}{1} ,\dots,\ww{1}{n-1} , \xx{1}{1} } 
            + n^{n-2} \sum_{j=2}^n \sum_{i=1}^n \entropy{\ww{j}{1} ,\dots,\ww{j}{n-1} , \xx{j}{i} } \\
        &\; \geq n^{n-1} \entropy{\ww{1}{1} ,\dots,\ww{1}{n-1} , \xx{1}{1} } 
            + n^{n-2} \sum_{j=2}^n k (n^2 - n + 1) 
            & & \Comment{\eqref{eq:entropy_9}}
            \\
        &\; = n^{n-1} \entropy{\ww{1}{1} ,\dots,\ww{1}{n-1} , \xx{1}{1} } 
            + n^{n-2} \, k \, (n-1) (n^2 - n + 1)
  \end{align*}
  and so
  \begin{align*}
    \entropy{\ww{1}{1} ,\dots,\ww{1}{n-1} , \xx{1}{1} } 
      & \leq  k \left( \frac{ n^2 - n + 1 }{n}\right).
  \end{align*}

  Similarly, for each $i,j=1,\dots,n$, we have
  \begin{align}
    \entropy{\ww{i}{1} ,\dots,\ww{i}{n-1} , \xx{i}{j} } 
      & \leq  k \left( \frac{ n^2 - n + 1 }{n} \right)
      \label{eq:entropy_10}.
  \end{align}

  However, for each $i= 1,\dots,n$ we also have
  \begin{align*}
    k (n^2 - n + 1) 
      & \leq \sum_{j=1}^n \entropy{\ww{i}{1} ,\dots,\ww{i}{n-1} , \xx{i}{j} }
        & & \Comment{\eqref{eq:entropy_9}} \\
      & \leq \sum_{j = 1}^n k \left( \frac{ n^2 - n + 1 }{n} \right)
        & & \Comment{\eqref{eq:entropy_10}}\\
      & = k \, (n^2 - n + 1)
  \end{align*}
  and so
  for each $i,j=1,\dots,n$,
  \begin{align*}
    \entropy{\ww{i}{1} ,\dots,\ww{i}{n-1} , \xx{i}{j} }
      & =  k \left( \frac{ n^2 - n + 1 }{n} \right).
  \end{align*}

  The variables $\ww{i}{1} ,\dots,\ww{i}{n-1} , \xx{i}{j}$ are linear functions of the uniformly distributed messages,
  so by Lemma~\ref{lem:int_ent},
  $\entropy{\ww{i}{1} ,\dots,\ww{i}{n-1} , \xx{i}{j}}$ (with logarithms in base $|\F|$) is an integer.
  However, 
  $$\GCD{n}{n^2-n+1} = \GCD{n}{(n^2 - n + 1) - n(n-1)} = \GCD{n}{1} = 1$$ 
  so if $ k\left( \frac{ n^2 - n + 1}{n} \right)$
  is an integer, then we must have $\Div{n}{k}$.
\end{proof}

The following corollary demonstrates it is possible for a network 
to be scalar linearly solvable over a non-commutative ring
but not over any commutative rings,
which is, in fact, equivalent to a network being vector linearly solvable over some field
but not scalar linearly solvable over any field,
by Corollaries~\ref{cor:general_vector} and \ref{cor:general_scalar}.

\begin{corollary}
  For all integers $n \geq 2$, $k\ge 1$, and prime $p$,
  the Dim-$n$ Network has a scalar linear solution over a non-commutative ring of size $p^{k n^2}$
  but has no scalar linear solution over any commutative ring.
  \label{cor:non-comm-ring}
\end{corollary}
\begin{proof}
  If the Dim-$n$ Network were scalar linearly solvable over a commutative ring,
  then by Corollary~\ref{cor:general_scalar},
  the Dim-$n$ Network would also be scalar linearly solvable over some finite field.
  However, by Theorem~\ref{thm:dim-n_vector},
  the Dim-$n$ Network is not scalar linearly solvable over any finite field.

  By Theorem~\ref{thm:dim-n_vector},
  the Dim-$n$ Network has an $n$-dimensional vector linear solution over $\GF{p^k}$,
  so by Corollary~\ref{cor:same_ring}
  the Dim-$n$ Network has a linear solution over the ring $M_{n}(\GF{p^k})$.
\end{proof}

\begin{corollary}
  For each integer $n \geq 2$,
  the unique smallest-size ring over which the Dim-$n$ Network is
  scalar linearly solvable is
  the ring of all $n \times n$ matrices over $\GF{2}$.
  \label{cor:dim-n_smallest_ring}
\end{corollary}
\begin{proof}
  By Theorem~\ref{thm:dim-n_vector},
  the Dim-$n$ Network has an $n$-dimensional vector linear solution over $\GF{2}$,
  and by Corollary~\ref{cor:same_ring},
  the Dim-$n$ Network has a linear solution over the ring $M_{n}(\GF{2})$.

  Suppose the Dim-$n$ Network is scalar linearly solvable over a ring $R$
  such that $|R| \leq 2^{n^2}$.
  By Lemmas~\ref{lem:simple_rings} and \ref{lem:hom_to_simple} (a) (b)
  there exists a field $\F$, a positive integer $k$,
  and a surjective homomorphism $\phi:R \to M_k(\F)$
  such that the Dim-$n$ Network is scalar linearly solvable over $M_k(\F)$.
  By Corollary~\ref{cor:same_ring},
  this implies the Dim-$n$ Network has a $k$-dimensional vector linear solution over $\F$,
  which by Theorem~\ref{thm:dim-n_vector}, 
  implies $n$ divides $k$.
  Since $\phi$ is surjective, $|M_k(\F)| \leq |R|$.
  Hence we have 
  $$2^{n^2} \leq 2^{k^2} \leq |\F|^{k^2} = |M_k(\F)| \leq |R| \leq 2^{n^2}.$$
  Therefore $k = n$ and $\F = \GF{2}$.
  Since $|R| = |M_k(\F)|$ and $\phi$ is a surjective homomorphism,
  we have $R \cong M_n(\GF{2}).$
\end{proof}

\begin{example}
Setting $k=1$ and $p=n=2$ in Corollary~\ref{cor:non-comm-ring}
results in the M Network
(see Figure~\ref{fig:M-network})
having no scalar linear solution over any commutative ring
but having a scalar linear solution over a non-commutative ring of size $16$.
The non-commutative ring $M_2(\GF{2})$ consists of all $2\times 2$ binary matrices under ordinary
matrix addition and multiplication mod $2$. 
Denote the $16$ ring elements by:
\begin{align*}
R_{qrst} &= \left[ \begin{array}{cc} q & r \\ s & t \\ \end{array}\right] 
\ \ \ \ \ \ \ \ (q,r,s,t \in \{0,1\}).
\end{align*}
A scalar linear solution for the M Network over the non-commutative ring $M_2(\GF{2})$ \\
(i.e. where $A,B,C,D,E,F,G,H, W,X,Y,Z \in M_2(\GF{2})$)
is given by:
\begin{align*}
\text{Edge (1,3)}: A &= R_{1000} W + R_{0010} X & \ \ \ \text{Decode at node 6}: W &= R_{1000} A + R_{0010} E + R_{0000} D \\
\text{Edge (1,4)}: B &= R_{0100} W + R_{0001} X & \ \ \                          Y &= R_{0000} A + R_{0001} E + R_{1000} D \\
\text{Edge (2,4)}: C &= R_{0100} Y + R_{0001} Z & \ \ \ \text{Decode at node 7}: W &= R_{1000} A + R_{0010} F + R_{0000} D \\
\text{Edge (2,5)}: D &= R_{1000} Y + R_{0010} Z & \ \ \                          Z &= R_{0000} A + R_{0001} F + R_{0100} D \\
\text{Edge (4,6)}: E &= R_{1000} B + R_{0010} C & \ \ \ \text{Decode at node 8}: X &= R_{0100} A + R_{0010} G + R_{0000} D \\
\text{Edge (4,7)}: F &= R_{1000} B + R_{0001} C & \ \ \                          Y &= R_{0000} A + R_{0001} G + R_{1000} D \\
\text{Edge (4,8)}: G &= R_{0100} B + R_{0010} C & \ \ \ \text{Decode at node 9}: X &= R_{0100} A + R_{0010} H + R_{0000} D \\
\text{Edge (4,9)}: H &= R_{0100} B + R_{0001} C & \ \ \                          Z &= R_{0000} A + R_{0001} H + R_{0100} D,
\end{align*}
where the out-edges of nodes with a single in-edge
each carry the symbol on the in-edge,
that is, each receiver directly receives the edge symbols $A$ and $D$
from the nodes $3$ and $5$, respectively.

We also note that if the messages and edge symbols of the M Network
are $2$-dimensional vectors over $\GF{2}$, instead of $2 \times 2$ binary matrices, 
then a small modification of the linear code described above
provides the $2$-dimensional vector linear solution over $\GF{2}$ given in 
\cite{Medard-NonMulticast}.
This agrees with Corollary~\ref{cor:same_ring}.

\psfrag{w,x}{$W,X$}
\psfrag{y,z}{$Y,Z$}
\psfrag{w,y}{$W,Y$}
\psfrag{w,z}{$W,Z$}
\psfrag{x,y}{$X,Y$}
\psfrag{x,z}{$X,Z$}
\psfrag{a}{$A$}
\psfrag{b}{$B$}
\psfrag{c}{$C$}
\psfrag{d}{$D$}
\psfrag{e}{$E$}
\psfrag{f}{$F$}
\psfrag{g}{$G$}
\psfrag{h}{$H$}
\begin{figure}[h]
  \begin{center}
    \leavevmode
    \hbox{\epsfxsize=.65\textwidth\epsffile{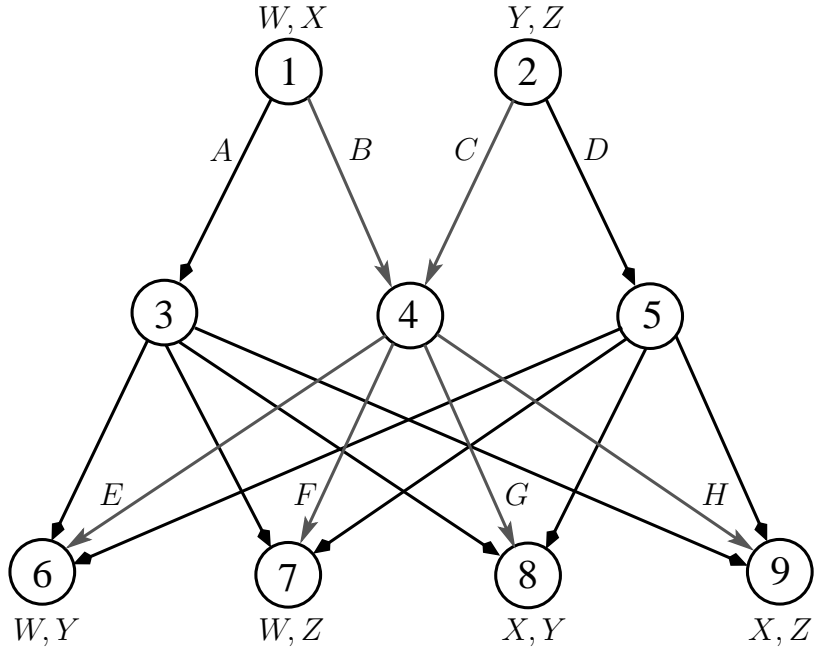}}
  \end{center}
  \caption{The M network has a non-commutative scalar linear solution.
           The messages $W,X,Y,Z$ take values in $M_2(\GF{2})$.
           The variables $A,B,C,D,E,F,G,H$ also take values in $M_2(\GF{2})$ and
           represent the symbols carried on the 8 indicated edges.
  }
\label{fig:M-network}
\end{figure}

\label{ex:16-noncommutative}
\end{example}

The bound in the following theorem is tight via Example~\ref{ex:16-noncommutative}.

\begin{theorem}
  If a network is scalar linearly solvable over some non-commutative ring $R$,
  but not over any commutative rings,
  then $|R| \geq 16$.
  \label{thm:non-comm_16}
\end{theorem}
\begin{proof}
  Suppose network $\Network$ is scalar linearly solvable over some non-commutative ring $R$
  and is not linearly solvable over any commutative ring.
  By Theorem~\ref{thm:R_dom_by_simple},
  there exists a positive integer $k$ and a field $\F$ such that
  $\Network$ has a linear solution over $M_k(\F)$
  and $|R| \geq |M_k(\F)|$.
  If $k = 1$, then $\Network$ is linearly solvable over a field,
  which contradicts the assumption that $\Network$ is not linearly solvable over any commutative ring.
  So $k \geq 2$,
  which implies $|R| \geq |M_k(\F)| = |\F|^{k^2} \ge 2^4 = 16$. 
\end{proof}

\clearpage
\section{Modules with the same alphabet size}
\label{sec:mods}

The following theorem demonstrates that there exists a network that is linearly solvable over a module
of size $p^k$
but not over any ring of size $p^k$.
\begin{theorem}
  For each integer $k \geq 2$ and prime $p$,
  the Dim-$k$ Network has a $k$-dimensional vector linear solution over the field $\GF{p}$
  but is not scalar linearly solvable over any ring of size $p^k$.
  \label{thm:module-beats-rings}
\end{theorem}
\begin{proof}
  By Theorem~\ref{thm:dim-n_vector}, the Dim-$k$ Network has a $k$-dimensional vector linear solution over $\GF{p}$.
  Let $R$ be a ring of size $p^k$ and suppose the Dim-$k$ Network has a scalar linear solution over $R$.
  By Lemmas~\ref{lem:simple_rings} and \ref{lem:hom_to_simple} (b) (c),
  there exists a field $\F$ and a positive integer $n$
  such that any network that is scalar linearly solvable over $R$ is also scalar linearly solvable over $M_n(\F)$
  and $|\F|^{n^2}$ divides $p^k$.
  Hence $\F$ is a field of characteristic $p$ and $n^2 \leq k$.
  
  Since the Dim-$k$ Network is scalar linearly solvable over $R$,
  the Dim-$k$ Network is scalar linearly solvable over the ring $M_n(\F)$.
  By Corollary~\ref{cor:same_ring}, this implies the Dim-$k$ Network 
  has an $n$-dimensional vector linear solution over $\F$,
  which by Theorem~\ref{thm:dim-n_vector} implies
  $\Div{k}{n}$.
  However, this contradicts the fact that $n^2 \leq k$.
  Thus, no such ring $R$ exists.
\end{proof}

\subsection{Commutative rings}

Both a scalar linear code over a ring of size $p^k$
and a $k$-dimensional vector linear code
are linear codes over a module of size $p^k$.
We have already seen (in Theorem~\ref{thm:module-beats-rings})
that there exists a network with a $k$-dimensional vector linear solution over $\GF{p}$
yet with no scalar linear solutions over any ring of size $p^k$.
The main result of this section (Theorem~\ref{thm:commutative_implies_vector})
will show that
any network that is scalar linearly solvable over a commutative ring of size $p^k$
must also have a $k$-dimensional vector linear solution over $\GF{p}$.

The following lemma was proved in \PartOne{}
(in \cite[Lemmas \PartOneLemmaDirectProduct{} and \PartOnePartDomRing]{Connelly-Ring1})
and will be used in what follows.

\begin{lemma}
  For each prime $p$ and positive integer $k$,
  if a network $\Network$ has a scalar linear solution over some commutative ring of size $p^k$,
  then there exists an integer partition $(n_1,\dots,n_r)$ of $k$ such that
  $\Network$ is scalar linearly solvable over each of the fields $\GF{p^{n_1}},\dots,\GF{p^{n_r}}$.
  \label{lem:dom_by_max}
\end{lemma}

The following standard result on rings will be used in later proofs.
\begin{lemma}{\cite[Theorem I.1]{McDonald-FiniteRings}}
  Every finite ring is isomorphic to a direct product of rings of prime power sizes.
  \label{lem:prime-power-sizes}
\end{lemma}

\begin{theorem}
  Let $m$ be a positive integer with prime factorization
  $m = \PrimeFact{t}$.
  If a network $\Network$ has a scalar linear solution over some commutative ring of size $m$,
  then the following hold:
  \begin{itemize}
    \item[(a)] For each $i = 1,\dots,t$,
    network $\Network$ has a $k_i$-dimensional vector linear solution over $\GF{p_i}$.
    \item[(b)] Network $\Network$ has a linear solution over 
    the 

    $M_{k_{1}}(\GF{p_{1}}) \times \cdots \times M_{k_{t}}(\GF{p_{t}})$-module
    $\GF{p_{1}}^{k_{1}} \times \cdots \times \GF{p_{t}}^{k_{t}}$.
  \end{itemize}
  \label{thm:commutative_implies_vector}
\end{theorem}
\begin{proof}
  Suppose $\Network$ is scalar linearly solvable over a commutative ring $R$ of size $m$.
  By Lemma~\ref{lem:prime-power-sizes},
  there exist rings $R_1,\dots,R_t$
  such that $R \cong R_1 \times \cdots \times R_t$
  and $|R_i| = p_i^{k_i}$ for all $i$.
  
  Let $i \in \{1,\dots,t\}$.
  Since the projection mapping from $R$ to $R_i$ is a surjective homomorphism,
  by Corollary~\ref{cor:ModHomomorphism},
  network $\Network$ is scalar linearly solvable over $R_i$.
  Then by Lemma~\ref{lem:dom_by_max},
  there exists an integer partition $(n_1,\dots,n_r)$ of $k_i$ such that
  $\Network$ is scalar linearly solvable over each of the fields $\GF{p_i^{n_1}},\dots,\GF{p_i^{n_r}}$.
  By Lemma~\ref{lem:kn-vector},
  this implies that $\Network$ has an $n_j$-dimensional vector linear solution over $\GF{p_i}$
  for each $j = 1,\dots,r$.
  However, by Lemma~\ref{lem:dim_sum},
  this then implies that $\Network$ has a $k_i = (n_1+\cdots+n_r)$-dimensional vector linear solution over $\GF{p_i}$.

  Hence, for all $i \in \{1,\dots,t\}$,
  a Cartesian product code formed from the
  $k_{i}$-dimensional vector linear solutions over $\GF{p_{i}}$
  gives a linear solution to $\Network$ over the described module.
\end{proof}

In \PartOne, we showed 
(in \cite[Theorems~\PartOneTheoremBestRingNetwork{} and \PartOneTheoremPFiveSeven]{Connelly-Ring1})
that with respect to ring domination for scalar linear coding,
some ring sizes give rise to multiple maximal commutative rings 
whereas other ring sizes yield only a single unique maximal commutative ring.
If there is just one maximal commutative ring of size $m$,
then every network that is linearly solvable over some commutative ring of size $m$
is also linearly solvable over the maximal ring.
In contrast, if there are multiple maximal commutative rings of size $m$, 
then for any commutative ring $R$ of size $m$,
there is always a different commutative ring $S$ also of size $m$,
such that some network is scalar linearly solvable over $S$
but not over $R$.
Thus, in this sense, there is no ``best'' commutative ring of a given size.

However,
by Theorem~\ref{thm:commutative_implies_vector} (b),
if a network has a linear solution over some commutative ring of size $m = \PrimeFact{t}$,
then it has a linear solution over the $M_{k_{1}}(\GF{p_{1}}) \times \cdots \times M_{k_{t}}(\GF{p_{t}})$-module
$\GF{p_{1}}^{k_{1}} \times \cdots \times \GF{p_{t}}^{k_{t}}$,
which also has size $m$.
In fact, we showed (in Theorem \ref{thm:module-beats-rings}) 
that when $m = p^k$, the converse is not true.
So in this sense, $k$-dimensional vector linear codes over $\GF{p}$ are strictly ``better'' than 
scalar linear codes over commutative rings of size $p^k$.

\subsection{Non-commutative rings}

This section generalizes the results of Theorem~\ref{thm:commutative_implies_vector}
to (not necessarily commutative) rings of size $m$ with prime factor multiplicity less than or equal to $6$.
In order to do so, we first will prove some intermediate results and consider special cases.

The following lemma was proved in \PartOne{}
(in \cite[Theorem \PartOneTheoremPFiveSeven]{Connelly-Ring1})
and will be used in what follows.

\begin{lemma}
  For each $k \in \{1,2,3,4,6\}$ and prime $p$,
  if a network is scalar linearly solvable over some commutative ring of size $p^k$,
  then it is scalar linearly solvable over $\GF{p^k}$.
\label{lem:p12346}
\end{lemma}

Lemma~\ref{lem:ring_size_p123} characterizes the non-commutative rings of prime-power size
whose multiplicity is at most three.

\begin{lemma}{\cite[pp. 512--513]{Eldridge-Orders}}
  For each prime $p$,
  all rings of size $p$ and of size $p^2$ are commutative,
  and the ring of all upper-triangular $2 \times 2$ matrices over $\GF{p}$
  is the only non-commutative ring of size $p^3$.
  \label{lem:ring_size_p123}
\end{lemma}
%

We remark that there exist rings of size $p$ and $p^2$ without identity.
For example, the set
$\left\{0,2,4,6 \right\}$
with mod $8$ addition and multiplication
satisfies all of the properties of a ring except 
there is no multiplicative identity.
However, such rings (sometimes called ``rngs'')
do not appear to be practical for linear network coding,
as receivers must recover their demands from linear combinations of their inputs.

For example, 
consider the trivial network shown in Figure~\ref{fig:trivial}
consisting of a single message $x$ emitted by a source directly 
connected by a single edge to a receiver demanding message $x$.
The only possible linear functions that can be carried on the edge
are of the form $cx$ for some fixed $c \in \{0,2,4,6\}$.
However, no matter what the choice of $c$ is,
the messages $0$ and $4$ always get received as $0$ mod $8$,
so the receiver cannot uniquely determine $x$ in general.
Thus, there is no linear solution for the network over this ring 
(with no multiplicative identity).
A similar issue arises for the set $\{0,2\}$
with mod $4$ addition and multiplication,
which also satisfies all of the properties of a ring except
there is no multiplicative identity.

\psfrag{x}{$x$}
\begin{figure}[h]
  \begin{center}
    \leavevmode
    \hbox{\epsfxsize=.35\textwidth\epsffile{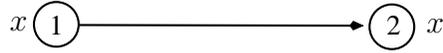}}
  \end{center}
  \caption{A trivial network with one message $x$ that is demanded by the receiver. }
\label{fig:trivial}
\end{figure}

\begin{lemma}
  For each prime $p$,
  if a network is scalar linearly solvable over some ring of size $p^2$,
  then it is a scalar linearly solvable over $\GF{p^2}$.
  \label{lem:p2}
\end{lemma}
\begin{proof}
  By Lemma~\ref{lem:ring_size_p123},
  every ring of size $p^2$ is commutative,
  and by Lemma~\ref{lem:p12346},
  every network that is scalar linearly solvable over some commutative ring of size $p^2$
  has a scalar linear solution over $\GF{p^2}$.
\end{proof}

\newcommand{\MatBig}[1]{ 
\left[ \begin{array}{ccc}
    #1_{1,1} & \cdots & #1_{1,k} \\
      & \ddots & \vdots \\
    \text{\LARGE0} &  & #1_{k,k}
  \end{array} \right]}

By Lemma~\ref{lem:ring_size_p123},
the smallest non-commutative ring is the ring of the $8$ binary upper-triangular $2\times 2$ matrices.
As a special case of the following lemma,
any network that is scalar linearly solvable over this ring
must also have a scalar linear solution over $\GF{2}$.

\begin{lemma}
  For each finite field $\F$ and integer $k \geq 2$,
  any network that is scalar linearly solvable over the ring of upper-triangular $k \times k$ matrices over $\F$
  is also scalar linearly solvable over $\F$.
  \label{lem:p3_matrices}
\end{lemma}
\begin{proof}
  Let $R$ be the ring of upper-triangular $k \times k$ matrices with entries in $\F$
  and let
  $\phi: R \to \F$ be given by
  $$\phi\left(
  \MatBig{a} \right)
  = a_{1,1}.$$
  Then $\phi$ is clearly surjective and preserves identities,
  and for any $A,B \in R$,
  \begin{align*}
  \phi(A + B)
  &= a_{1,1} + b_{1,1} 
  = \phi(A) + \phi(B)
  \\
  \phi(A B)
  &= a_{1,1}\, b_{1,1} 
  = \phi(A) \phi(B).
  \end{align*}
  Thus $\phi$ is a surjective homomorphism,
  so by Corollary~\ref{cor:ModHomomorphism},
  any network that is scalar linearly solvable over $R$
  is scalar linearly solvable over $\F$.
\end{proof}

\begin{lemma}
  For each prime $p$,
  if a network is scalar linearly solvable over some ring of size $p^3$,
  then it is scalar linearly solvable over $\GF{p^3}$.
  \label{lem:p3}
\end{lemma}
\begin{proof}
  By Lemma~\ref{lem:ring_size_p123},
  the only non-commutative ring of size $p^3$ is the ring of upper triangular matrices with entries in $\GF{p}$,
  and by Lemma~\ref{lem:p3_matrices},
  any network that is scalar linearly solvable over this ring
  is also scalar linearly solvable over $\GF{p}$.
  Since $\GF{p}$ is a subring of $\GF{p^3}$,
  any network that is scalar linearly solvable over $\GF{p}$
  is scalar linearly solvable over $\GF{p^3}$.

  By Lemma~\ref{lem:p12346},
  every network that is scalar linearly solvable over some commutative ring of size $p^3$
  has a scalar linear solution over $\GF{p^3}$.
\end{proof}

The following three lemmas are proved in the Appendix.

\begin{lemma}
  For each prime $p$,
  if a network is scalar linearly solvable over some ring of size $p^4$,
  then it is scalar linearly solvable over at least one of the rings $\GF{p^4}$ or $M_2(\GF{p})$.
  \label{lem:p4}
\end{lemma}

\begin{lemma}
  For each prime $p$,
  if a network is scalar linearly solvable over some ring of size $p^5$,
  then it is scalar linearly solvable over at least one of the rings $\GF{p^5}$ or $\GF{p^3} \times \GF{p^2}$.
  \label{lem:p5}
\end{lemma}

\begin{lemma}
  For each prime $p$,
  if a network is scalar linearly solvable over some ring of size $p^6$,
  then it is scalar linearly solvable over $\GF{p^6}$.
  \label{lem:p6}
\end{lemma}

Theorem~\ref{thm:p23456}
is a generalization of Lemma~\ref{lem:p12346} to scalar linear codes over non-commutative rings.
Extending Theorem~\ref{thm:p23456} to $|R| = p^k$ for $k \geq 7$ is left as an open problem.

\begin{theorem}
  Let $p$ be a prime, and suppose $\Network$ is scalar linearly solvable over a ring $R$.
  \begin{itemize}[leftmargin=0.65cm]
    \item[(a)] If $|R| = p^2$, then $\Network$ is scalar linearly solvable over $\GF{p^2}$.
    \item[(b)] If $|R| = p^3$, then $\Network$ is scalar linearly solvable over $\GF{p^3}$.
    \item[(c)] If $|R| = p^4$, then $\Network$ is scalar linearly solvable over at least one of $\GF{p^4}$ or $M_2(\GF{p})$.
    \item[(d)] If $|R| = p^5$, then $\Network$ is scalar linearly solvable over at least one of $\GF{p^5}$ or $\GF{p^3} \times \GF{p^2}$.
    \item[(e)] If $|R| = p^6$, then $\Network$ is scalar linearly solvable over $\GF{p^6}$.
  \end{itemize}
  \label{thm:p23456}
\end{theorem}
\begin{proof}
  This follows immediately from Lemmas~\ref{lem:p2}, \ref{lem:p3}, \ref{lem:p4}, \ref{lem:p5}, and \ref{lem:p6}.
\end{proof}

We also note that by 
Corollary~\ref{cor:choose-two-ring},
the $(p^4+1)$-Choose-Two Network is scalar linearly solvable over $\GF{p^4}$ but not over $M_2(\GF{p})$
and the $(p^5+1)$-Choose-Two Network is scalar linearly solvable over $\GF{p^5}$ but not over $\GF{p^3} \times \GF{p^2}$.
By Corollary~\ref{cor:non-comm-ring},
the Dim-$2$ Network is scalar linearly solvable over $M_2(\GF{p})$ but not over $\GF{p^4}$.
We showed in \PartOne{} \cite[Theorem \PartOneTheoremOnlyThirtyTwo]{Connelly-Ring1}
that there exists a network that is scalar linearly solvable over $\GF{p^3} \times \GF{p^2}$
but not over $\GF{p^5}$.
Hence it is necessary to include both rings in (c) and (d) in Theorem~\ref{thm:p23456}.

\begin{corollary}
  Let $p$ be a prime and $k \in \{2,3,4,5,6\}$, 
  and suppose $\Network$ is scalar linearly solvable over a ring of size $p^k$.
  Then $\Network$ has a $k$-dimensional vector linear solution over $\GF{p}$.
  \label{cor:p23456}
\end{corollary}
\begin{proof}
  If $k \in \{2,3,5,6\}$,
  then by Theorem~\ref{thm:p23456},
  $\Network$ has a scalar linear solution over a commutative ring of size $p^k$,
  so by Theorem~\ref{thm:commutative_implies_vector},
  $\Network$ has a $k$-dimensional vector linear solution over $\GF{p}$.

  Now suppose $k=4$.
  If $\Network$ is scalar linearly solvable over $\GF{p^4}$,
  then by Lemma~\ref{lem:kn-vector},
  $\Network$ has a $4$-dimensional vector linear solution over $\GF{p}$.
  If $\Network$ is not scalar linearly solvable over $\GF{p^4}$,
  then by Theorem~\ref{thm:p23456} (c),
  $\Network$ must be scalar linearly solvable over $M_2(\GF{p})$,
  so by Corollary~\ref{cor:same_ring},
  $\Network$ has a $2$-dimensional vector linear solution over $\GF{p}$,
  in which case
  $\Network$ also has a $4$-dimensional vector linear solution over $\GF{p}$
  by Lemma~\ref{lem:dim_sum}.
\end{proof}

Theorem~\ref{thm:p6_implies_vector} generalizes the results of Theorem~\ref{thm:commutative_implies_vector}
to rings of size $m$ with prime factor multiplicity less than or equal to $6$.

\begin{theorem}
  Let $m$ be a positive integer with prime factorization
  $m = \PrimeFact{t}$.
  If a network $\Network$ has a scalar linear solution over a ring of size $m$,
  then,
  for each $i = 1,\dots,t$ such that $k_i \leq 6$,
  network $\Network$ has a $k_i$-dimensional vector linear solution over $\GF{p_i}$.
\label{thm:p6_implies_vector}
\end{theorem}
\begin{proof}
  Suppose $\Network$ is scalar linearly solvable over a ring $R$ of size $m$.
  By Lemma~\ref{lem:prime-power-sizes},
  there exists rings $R_1,\dots,R_t$
  such that $R \cong R_1 \times \cdots \times R_t$
  and $|R_i| = p_i^{k_i}$ for all $i$.
  
  Now, let $i \in \{1,\dots,t\}$ and suppose $k_i \leq 6$.
  The projection mapping from $R$ to $R_i$
  is a surjective homomorphism,
  so by Corollary~\ref{cor:ModHomomorphism},
  network $\Network$ is scalar linearly solvable over $R_i$.
  Since $\Network$ is scalar linearly solvable over a ring of size $p_i^{k_i}$ where $k_i \leq 6$,
  by Corollary~\ref{cor:p23456},
  $\Network$ has a $k_i$-dimensional vector linear solution over $\GF{p_i}$.
\end{proof}

We leave as an open question whether
the restriction that $k_i \leq 6$ can be removed from the statement of Theorem~\ref{thm:p6_implies_vector}.
If this generalization is false, 
then for what primes $p$ and positive integers $k$
is it the case that there exists a network with
a scalar linear solution over a ring of size $p^k$
but with no $k$-dimensional vector linear solution over $\GF{p}$?
If such a ring and such a network do exist,
the ring must be non-commutative
and $k \geq 7$.
More generally, does there exist a network with a linear solution over some alphabet of size $p^k$
but with no $k$-dimensional vector linear solution over $\GF{p}$?

\clearpage
\section{Concluding Remarks}
\label{sec:complexity}

For each positive integer $k$ and prime $p$, we have shown 
the set
$$\{\Network \, : \, \Network \text{ has a scalar linear solution over some commutative ring of size } p^k  \}$$
is properly contained in
$$\{\Network \, : \, \Network \text{ has a $k$-dimensional vector linear solution over } \GF{p}\}.$$

So in this sense, $k$-dimensional vector linear codes over $\GF{p}$
may be advantageous compared to scalar linear codes over commutative rings of the same size $p^k$.
In addition, there are more
$k$-dimensional linear functions over $\GF{p}$
than there over a commutative ring of size $p^k$.
Vector linear codes over fields are also optimal in the sense that
they minimize the alphabet size needed for a linear solution over a particular network.
On the other hand, the complexity of implementing vector linear codes is generally higher than 
for scalar linear codes over commutative rings of the same size.

\clearpage

\makeatletter\renewcommand{\@seccntformat}[1]{}\makeatother
\appendix
\section{Appendix}

The main purpose of this Appendix is to prove Lemmas~\ref{lem:p4}, \ref{lem:p5}, and \ref{lem:p6},
which are used in the proof of Theorem~\ref{thm:p23456}.
It is an open question whether Theorem~\ref{thm:p6_implies_vector}
can be extended to all finite rings.
The techniques presented in this section may additionally be useful for examining such questions.

Recall that a finite ring is simple if it has no proper two-sided ideals.
The \textit{radical} of a ring $R$ is the intersection of all its maximal left ideals.
The radical of a ring is a two-sided ideal.
A finite ring $R$ with radical $J$ is said to be:
\begin{itemize}
  \item \textit{local}
  \footnote{If $R$ is a local commutative ring, then $R$ has a single maximal ideal,
  which corresponds to our definition of a commutative local ring in \PartOne.} 
  if $R/J$ is a field.
  \item \textit{semi-local} if $R/J$ is simple,
  or equivalently $R$ is isomorphic to a matrix over a local ring
  (e.g. \cite[p. 162]{McDonald-FiniteRings}).
  \item \textit{semi-simple} if $R$ is isomorphic to a direct product of simple rings (matrix rings over fields)
  or equivalently, $J = \{0\}$ (e.g. \cite[pp. 75, 128]{McDonald-FiniteRings}).
\end{itemize}

The following lemma is a result on local rings that will be used in later proofs.
\begin{lemma}
  Let $p$ be a prime, $k$ a positive integer,
  and $R$ a semi-local ring of size $p^k$.
  Then there exists a unique local ring $S$ and positive integers $r,s,t$
  such that the following hold:
  \begin{itemize}
  \item[(a)] \cite[Theorem VIII.26]{McDonald-FiniteRings} $R \cong M_r(S)$ 
  \item[(b)] \cite[Theorem 6.1.2]{Bini-FCR} $|S| = p^s$
  \item[(c)] \cite[Theorem 6.1.2]{Bini-FCR} $\GF{p^t} \cong S/J$, where $J$ is the radical of $S$
  and $\Div{t}{s}$.
  \end{itemize}
  \label{lem:local_rings}
\end{lemma}
%

As an example,
let $p$ be a prime and let $r,s$ be positive integers.
Then $M_r(\Z_{p^s})$ is a semi-local ring,
since $\Z_{p^s}$ is a local ring.
We also remark that in Lemma~\ref{lem:local_rings}, 
if $R$ is itself local, 
then $S \cong R$.

The following lemmas are results on semi-simple rings and the radicals of rings.

\begin{lemma}{\cite[Proposition IV.6, Theorem VIII.4]{McDonald-FiniteRings})}
  Let $R$ be a finite ring with radical $J$.
  Then there exist fields $\F_1,\dots,\F_s$
  and positive integers $r_1,\dots,r_s$ such that
  $$R/J \cong M_{r_1}(\F_1) \times \cdots \times M_{r_s}(\F_s).$$
  \label{lem:semisimple_radical}
\end{lemma}

\begin{lemma}
  Let $R$ be a finite ring with radical $J$,
  and suppose 
  $$R/J \cong M_{r_1}(\F_1) \times \cdots \times M_{r_s}(\F_s)$$
  for some fields $\F_1,\dots,\F_s$ and positive integers $r_1,\dots,r_s$.
  If a network is scalar linearly solvable over $R$,
  then it is also scalar linearly solvable over each of the rings
  $M_{r_1}(\F_1),\dots, M_{r_s}(\F_s)$.
  \label{lem:p6_semi1}
\end{lemma}
\begin{proof}
  By Lemma~\ref{lem:surjective_homomorphism},
  there exists a surjective homomorphism $\phi: R \to R/J$.
  Let $i \in \{1,\dots,s\}$.
  Then the projection mapping $\psi_{i}: R/J \to M_{r_i}(\F_i)$
  is a surjective homomorphism.
  Hence the composition of mappings 
  $\psi_{i} \circ \phi: R \to M_{r_i}(\F_i)$
  is a surjective homomorphism.
  Thus by Corollary~\ref{cor:ModHomomorphism}, any network with a scalar linear solution over $R$
  has a scalar linear solution over the ring $M_{r_i}(\F_i)$.
\end{proof}

The following is an enumeration of semi-simple rings that we will reference in upcoming proofs.
For each prime $p$, 
it can be verified that the rings given in \eqref{eq:r1}--\eqref{eq:r111111}
are all of the semi-simple rings of sizes $p,p^2,p^3,p^4,p^5,$ or $p^6$ (up to isomorphism).

\begin{align}
  \text{Size } p \;: \; \; 
  &\GF{p} \label{eq:r1}\\
  \notag \\
  \text{Size } p^2: \; \; 
  &\GF{p^2} \label{eq:r2}\\
  &\GF{p} \times \GF{p} \label{eq:r11} \\
  \notag \\
  \text{Size } p^3: \; \; 
  &\GF{p^3} \label{eq:r3} \\
  &\GF{p^2} \times \GF{p} \label{eq:r21} \\
  &\GF{p} \times \GF{p} \times \GF{p} \label{eq:r111} \\
  \notag \\
  \text{Size } p^4: \; \; 
  &M_2(\GF{p}) \label{eq:M4}\\
  &\GF{p^4} \label{eq:r4}\\
  &\GF{p^3} \times \GF{p} \label{eq:r31} \\
  &\GF{p^2} \times \GF{p^2} \label{eq:r22} \\
  &\GF{p^2} \times \GF{p} \times \GF{p} \label{eq:r211} \\
  &\GF{p} \times \GF{p} \times \GF{p} \times \GF{p} \label{eq:r1111} \\
  \notag \\
  \text{Size } p^5: \; \; 
  & \GF{p^5} \label{eq:r5} \\
  &M_2(\GF{p}) \times \GF{p} \label{eq:M41} \\
  &\GF{p^4} \times \GF{p} \label{eq:r41} \\
  &\GF{p^3} \times \GF{p^2} \label{eq:r32} \\
  &\GF{p^3} \times \GF{p} \times \GF{p} \label{eq:r311} \\
  &\GF{p^2} \times \GF{p^2} \times \GF{p} \label{eq:r221} \\
  &\GF{p^2} \times \GF{p} \times \GF{p} \times \GF{p} \label{eq:r2111} \\
  &\GF{p}\times \GF{p} \times \GF{p}\times \GF{p}\times \GF{p} \label{eq:r11111} 
\end{align}
\begin{align}
\text{Size } p^6: \; \; 
  &\GF{p^6} \label{eq:r6} \\
  &\GF{p^5} \times \GF{p} \label{eq:r51} \\
  &M_2(\GF{p}) \times \GF{p^2} \label{eq:M42} \\
  &\GF{p^4} \times \GF{p^2} \label{eq:r42} \\
  &M_2(\GF{p}) \times \GF{p} \times \GF{p} \label{eq:M411} \\
  &\GF{p^4} \times \GF{p} \times \GF{p} \label{eq:r411} \\
  &\GF{p^3} \times \GF{p^3} \label{eq:r33} \\
  &\GF{p^3} \times \GF{p^2} \times \GF{p} \label{eq:r321} \\
  &\GF{p^3} \times \GF{p} \times \GF{p} \times \GF{p} \label{eq:r3111} \\
  &\GF{p^2} \times \GF{p^2} \times \GF{p^2} \label{eq:r222} \\
  &\GF{p^2} \times \GF{p^2} \times \GF{p} \times \GF{p} \label{eq:r2211} \\
  &\GF{p^2} \times \GF{p} \times \GF{p} \times \GF{p} \times \GF{p} \label{eq:r21111} \\
  &\GF{p} \times \GF{p} \times \GF{p} \times \GF{p} \times \GF{p} \times \GF{p} \label{eq:r111111}
\end{align}

We now prove Lemmas~\ref{lem:p4}, \ref{lem:p5}, and \ref{lem:p6}.


\begin{proof}[Proof of Lemma~\ref{lem:p4}]
  Let $R$ be a ring of size $p^4$ with radical $J$, and suppose $\Network$ is scalar linearly solvable over $R$.
  Then $|R/J| \in \{p,p^2,p^3,p^4\}$, 
  so by Lemma~\ref{lem:semisimple_radical}, 
  $R/J$ is isomorphic to one of the rings in \eqref{eq:r1}--\eqref{eq:r1111}.

  If $R/J$ is isomorphic to any of these rings except those in \eqref{eq:r3} and \eqref{eq:M4}, 
  then by Lemma~\ref{lem:p6_semi1},
  $\Network$ is also scalar linearly solvable over at least one of $\GF{p}$, $\GF{p^2}$, or $\GF{p^4}$.
  Since $\GF{p}$ and $\GF{p^2}$ are both subrings of $\GF{p^4}$,
  in these cases, $\Network$ is also scalar linearly solvable over $\GF{p^4}$.
  
  If $R/J$ is isomorphic to the ring in \eqref{eq:M4},
  then by Lemma~\ref{lem:p6_semi1},
  $\Network$ is also scalar linearly solvable over $M_2(\GF{p})$.
  It follows from Lemma~\ref{lem:local_rings}
  that $R/J$ is not isomorphic to the ring in \eqref{eq:r3}.
\end{proof}


\begin{proof}[Proof of Lemma~\ref{lem:p5}]
  Let $R$ be a ring of size $p^5$ with radical $J$, and 
  suppose $\Network$ is scalar linearly solvable over $R$.
  Then $|R/J| \in \{p,p^2,p^3,p^4,p^5\}$,
  so by Lemma~\ref{lem:semisimple_radical},
  $R/J$ must be isomorphic to one of the rings in \eqref{eq:r1}--\eqref{eq:r11111}.

  If $R/J$ is isomorphic to one of the rings in \eqref{eq:M4}--\eqref{eq:r1111}
  (i.e. $|R/J| = p^4$),
  then $|J| = p$.
  Since $(J,+)$ is an $R$-module
  and $\Network$ has a linear solution over the faithful module $\Module{R}{R}$,
  by Lemma~\ref{lem:same_ring}, $\Network$ has a linear solution over $\Module{J}{R}$.
  By Theorem~\ref{thm:min_module},
  this implies $\Network$ has a scalar linear solution over $\GF{p}$.
  %
  Since $\GF{p}$ is a subring of $\GF{p^5}$,
  in these cases, $\Network$ also has a scalar linear solution over $\GF{p^5}$.

  It follows from Lemma~\ref{lem:local_rings}
  that $R/J$ is not isomorphic to either of the rings in
  \eqref{eq:r2} or \eqref{eq:r3}.
  If $R/J$ is isomorphic to the ring in \eqref{eq:r32}, 
  then by Lemma~\ref{lem:p6_semi1},
  $\Network$ is scalar linearly solvable over $\GF{p^3} \times \GF{p^2}$.

  If $R/J$ is isomorphic to any of the remaining cases,
  then by Lemma~\ref{lem:p6_semi1},
  network $\Network$ is scalar linearly solvable over either $\GF{p}$ or $\GF{p^5}$.
  Since $\GF{p}$ is a subring of $\GF{p^5}$,
  in these cases, $\Network$ also has a scalar linear solution over $\GF{p^5}$.
\end{proof}


\begin{proof}[Proof of Lemma~\ref{lem:p6}]
  Let $R$ be a ring of size $p^6$ with radical $J$, and suppose $\Network$ is scalar linearly solvable over $R$.
  Then $|R/J| \in \{p,p^2,p^3,p^4,p^5,p^6\}$,
  so by Lemma~\ref{lem:semisimple_radical},
  $R/J$ must be isomorphic to one of the rings in \eqref{eq:r1}--\eqref{eq:r111111}.
  It follows from Lemma~\ref{lem:local_rings}
  that $R/J$ is not isomorphic to any of the rings in \eqref{eq:M4}, \eqref{eq:r4}, or \eqref{eq:r5}.

  If $R/J$ is isomorphic to any of the remaining cases, 
  then it follows from Lemma~\ref{lem:p6_semi1}
  that $\Network$ is scalar linearly solvable over $\GF{p^n}$
  for some $n \in \{1,2,3,6\}$.
  Since $\Div{n}{6}$, $\GF{p^n}$ is a subring of $\GF{p^6}$,
  which implies 
  $\Network$ is scalar linearly solvable over $\GF{p^6}$.
\end{proof}

\end{document}